\documentclass[conference]{IEEEtran}
\IEEEoverridecommandlockouts
\usepackage{cite}
\usepackage{amsmath,amssymb,amsfonts}
\usepackage{algorithm}
\usepackage{algorithmic}
\usepackage{graphicx}
\usepackage{textcomp}
\usepackage{xcolor}
\usepackage{soul}
\usepackage{url}
\usepackage{subcaption}
\usepackage{placeins}

\usepackage{xspace}
\usepackage{makecell}
\usepackage{etoolbox}
\usepackage{amsthm}
\usepackage{bbm}
\usepackage{stackengine}[2013-10-15]
\usepackage{booktabs}   
\usepackage{multirow}   
\usepackage{makecell}   

\newtheorem{lemma}{Lemma}
\newtheorem{corollary}{Corollary}
\theoremstyle{definition}

\newcommand{\argminimal}[1]{\ensuremath{\underset{#1}{\arg\min}}\xspace}

\newcommand{\globalcostfunc}{\ensuremath{g}\xspace}

\newcommand{\unfairnessweight}{\ensuremath{\alpha}\xspace}
\newcommand{\localcostweight}{\ensuremath{\beta}\xspace}
\newcommand{\globalcostweight}{\ensuremath{\gamma}\xspace}

\newcommand{\iepos}{I-EPOS\xspace}

\def\BibTeX{{\rm B\kern-.05em{\sc i\kern-.025em b}\kern-.08em
    T\kern-.1667em\lower.7ex\hbox{E}\kern-.125emX}}
\begin{document}

\title{The Optimization Trilemma: \\Efficiency, Comfort and Fairness \\in Decentralized Multi-agent Coordination\\
}
\author{
	\IEEEauthorblockN{
		Jovan Nikol{\'i}c\IEEEauthorrefmark{1}, Maciej Krzysztof Zuziak\IEEEauthorrefmark{2}, Evangelos Pournaras\IEEEauthorrefmark{2}\IEEEauthorrefmark{3}
	}\\
	\IEEEauthorblockA{\IEEEauthorrefmark{1} Google, Zurich, 8002 Zurich, Switzerland. E-mail: jovannikolic@google.com}
	\IEEEauthorblockA{\IEEEauthorrefmark{2} School of Computer Science, University of Leeds, Leeds, UK. E-mail: \{M.Zuziak,E.Pournaras\}@leeds.ac.uk}
	\IEEEauthorblockA{\IEEEauthorrefmark{3} School of Energy Systems, LUT University, 53850 Lappeenranta, Finland}
}

\maketitle

\begin{abstract}
The problem of fair multi-agent coordination in decentralized settings is one of the most pressing challenges for building efficient collaborative systems. Resource allocation is based on optimized collective arrangements accounting for agents' needs. Such coordination should not only be computationally efficient but also account for fairness, i.e., equitable redistribution of costs incurred by all agents. Recent literature has proposed several algorithms that efficiently determine optimal plan combinations balancing system-wide efficiency and individual discomfort of agents in a centralized setting. However, these works do not address equitable resource optimization in fully decentralized scenarios, specifically, the optimized redistribution of discomfort among coordinating agents so that none experiences a discomfort level that could lead to loss of incentive or polarization that can disrupt planned operations. In this work, we study the problem of optimizing three objectives: (i) system-wide efficiency, (ii) individuals' comfort and (iii) fairness (i.e., balancing of incurred discomfort costs) in decentralized multi-agent coordination. We design a novel model to optimize those three orthogonal objectives, without any substantial increase in communication and computational overhead. Through experiments on two real-world datasets, we validate the model and demonstrate that it can achieve fairer optimization outcomes, while satisfying agents' preferences and system goals. 
\end{abstract}


\section{Introduction}\label{Section:Introduction}

Efficient use of resources~\cite{hollands2008} with active participation and cooperation of citizens in collective decision-making~\cite{deakin2013} are fundamental characteristics of sustainable complex systems, in which resource allocation is most effective when it accounts for the dynamically evolving needs and preferences of participants. Efficient resource consumption can be achieved via cooperation: prompt, impactful decisions resulting from collaboration among individuals that adapt to dynamic environmental conditions, saving resources and reducing costs. For example, in bike-sharing systems, costs of manual bike relocations can be reduced through cyclists' cooperation, by choosing the station to pick up and which to park at, based on others' choices. Electricity bills of households connected to a Smart Grid can be reduced by collaborative appliance scheduling during periods when electricity prices are lower, e.g., at night or when renewable energy availability is higher. 

The notion of \textit{efficient resource use} may differ between the individual and system perspectives, leading to discrepancies in the level of discomfort experienced by particular agents. For bike-sharing operators, efficiency means minimal manual bike relocation costs; from cyclists' perspectives, it means picking up a bike at the nearest bike station. In the Smart Grid scenario, scheduling appliances at times when it is unsuitable for residents is inefficient from their perspective, but it reduces the risk of a blackout. Such a discrepancy may lead to substantial unfairness in the system, in which optimization towards only the global objective results in some agents incurring higher participation costs than others. This can yield to rebound effects with dropouts that can also in turn influence system-wide efficiency.

Previous work on combinatorial planning in decentralized multi-agent settings has considered either single-objective~\cite{hinrichs2013, pournaras2010} or multi-objective~\cite{pournaras2018} scenarios. However, none of the works to date has considered fairness in the distribution of incurred costs across agents, so that no agent suffers disproportionate discomfort, or the system creates an imbalanced allocation that discourages some agents from participation. Since the system assumes collaboration among agents, those who are treated unfairly may lack a rational incentive to continue collaborating~\cite{zuziak_data_2023,Khalid2026}. On the other hand, even agents which are placed in an advantageous position may be unsatisfied with the outcome of the allocation, either because of the \textit{inequity aversion}~\cite{dannenberg_inequity_2007} or reciprocity principles \cite{fehr_fairness_2000}.

The minimization of unfairness in decentralized systems is further perplexed by the lack of access to complete information about the distribution of discomfort costs among agents, which can be reconstructed only from messages exchanged between them. In line with this, we design a multi-agent system that efficiently manages trade-offs in a distributed manner across multiple objectives: (i) \emph{system-wide efficiency} (global\footnote{The terms \textit{inefficiency} and \textit{global cost} are used interchangeably.} inefficiency cost), (ii) \emph{comfort} (local\footnote{Similarly, the terms \textit{discomfort} and \textit{local cost} are used interchangeably.} discomfort cost), and (iii) the \emph{fairness} of agents' selections, by reusing messages already passed by the agents to optimize the primary system-wide objective (i.e. efficiency), thereby significantly reducing communication bottleneck. Subsequently, we empirically analyze the impact of the proposed model on real-world datasets, demonstrating the effect of unequal distribution of discomfort costs among agents and how the proposed framework can steer the system towards more equitable optimization solutions.

\subsection{Contribution}

This work establishes the following contribution:
\begin{enumerate}
  	 \item A model for fairness that augments distributed multi-objective optimization, accounting for three, often orthogonal, goals: system-wide efficiency, comfort, and fairness in how discomfort is distributed. 
    \item A cost-effective application of the model for fairness to the collective learning approach of I-EPOS, the \emph{iterative Economic Planning and Optimized Selections}~\cite{pournaras2018,Pournaras2020}, expanding the optimization capabilities with negligible communication and computational cost. 
    \item New insights from an empirical evaluation of the framework using one synthetic and two real-world datasets on how the unequal distribution of agents' discomfort costs influence the system and how the proposed model can yield optimized solutions that are fairer for agents.
    \item An open-source software implementation\footnote{Available at \url{https://github.com/epournaras/EPOS} (last accessed: July 2026).} into the EPOS software release for further experimentation, research and applicability by the broader community.
\end{enumerate}

In the light of this contributions, we make the following advances to state of the art~\cite{zhang_fairness_2014, jiang_learning_2019, kumar_decaf_2025,Wang2026,Zimmer2021,Aloor2024,Chouaki2026}: (i) The proposed model can scale to thousands of agents, where the impact of inequalities is becoming profound. (ii) The communication and computational cost is negligible given that the optimization of fairness is incepted as a built-in element of existing general-purpose multi-agent coordination and therefore, no additional mechanism is required to bridge the information gap in decentralized settings. (iii) Fairness can be co-optimized jointly and adaptively with system-wide efficiency and individuals' comfort, providing a modular and reconfigurable multi-objective optimization for large-scale decentralized multi-agent systems.

\subsection{Structure of the Work}

The rest of the work is structured as follows: Section~\ref{Section:Related_Work} illustrates the most relevant approaches to single and multi-objective combinatorial optimization, including multi-agent approaches that account for fairness. Section~\ref{Section:Methodology} formulates the multi-agent combinatorial optimization problem, while Section~\ref{Section:Decentralised_Optimisation} introduces the model for optimizing fairness in decentralized multi-objective optimization. Section~\ref{Section:I-EPOS} shows how fairness minimization is incorporated into a state-of-the-art collective learning approach without additional communication cost. Section~\ref{Section:Experiments} illustrates the experimental evaluation of the proposed model and the trade-offs of optimizing the three orthogonal objectives. Finally, conclusions and future work are drawn in Section~\ref{Section:Conclusions}.

\section{Related Work}\label{Section:Related_Work}

The problem introduced in this work assumes an adaptive, decentralized optimization of three objectives that are often orthogonal~\cite{pournaras2018}: system-wide efficiency, individuals' comfort, and the fairness of comfort distribution. In this section, we illustrate multi-agent approaches to multi-objective combinatorial optimization (Subsection \ref{Related Works: Multi-agent Approaches to Combinatorial Optimisation}) and prior approaches to fairness optimization in multi-agent systems (Section \ref{Related Works: Fairness}).

\subsection{Decentralized Multi-agent Combinatorial Optimization}\label{Related Works: Multi-agent Approaches to Combinatorial Optimisation}

Solutions to many well-known combinatorial optimization problems can be obtained in a multi-agent setting. While multi-agent systems allow for autonomy of agents' computations that can run in parallel and on multiple machines, resulting in a shorter runtime and more efficient resource utilization, they also require additional synchronization schemes and message passing protocols that increase their complexity. For example, the multi-agent version of Particle Swarm Optimization with agents of different types~\cite{nouiri2018} is used to solve the flexible job-shop scheduling problem, which is NP-hard. The same problem was solved using the multi-agent version of the genetic algorithm coupled with Tabu search, a memory-based local search scheme that potentially accepts worse solutions if no improvements are available in near proximity~\cite{azzouz2012}. Community detection in complex networks is also NP-hard, and it is solved in a multi-agent setting using a genetic algorithm, assuming a lattice-like agent topology~\cite{li2016}. Some multi-agent approaches allow solving multiple different problems simultaneously, by abstracting scheduling and routing problems~\cite{martin2016} using ontologies and two types of agents: a launcher agent that converts the problem into an ontology and initializes others, and metaheuristic agents that implement metaheuristic solution search specified in the ontology. Finally, the power supply restoration problem in Smart Grids, another variation of the knapsack problem that is also NP-hard, is solved in a multi-agent setting via Lagrangian relaxation~\cite{agrawal2015}.

The majority of these solutions adopt a centralized approach, whereas only a few consider adaptive decentralized approaches. The Combinatorial Optimization Heuristic for Distributed Agents (COHDA)~\cite{hinrichs2013} can solve multi-agent combinatorial problems without assuming a centralized orchestrator, but it accounts only for a single objective. It considers an arbitrary network topology, in which each agent exchanges~\textit{configurations} with all adjacent agents. EPOS can optimize a single objective while leveraging the agents' tree hierarchy~\cite{pournaras2010}. The following work on I-EPOS extends this to multi-objective optimization, while agents rely on a tree hierarchy to efficiently coordinate their decision making~\cite{pournaras2018,Pournaras2020}. Our model can accommodate any decentralized combinatorial optimization algorithm designed for multiple objectives simultaneously. As a proof of concept, we demonstrate the applicability of our model to I-EPOS as an exemplary algorithm with outstanding performance and practical use. Earlier work has also shown the gap on optimizing fairness~\cite{pournaras2014m,Pournaras2019}.

\subsection{Fairness in Multi-Agent Optimization}
\label{Related Works: Fairness}

The problem of fairness in multi-agent systems dates back to a cake-cutting problem -- a common metaphor for the division of an heterogeneous divisible good~\cite{chen_truth_2013, dubins_how_1961}, where the good may be the allocation of resources, such as slots in the power grid schedule or the use of cycling stations. Fairness in such systems is a direct relaxation of the perfect rationality assumption~\cite{jong_fairness_2008}, where the agent can take into account choices that are not perfectly rational in the sake of satisfying some of their pre-defined normative principles. 

From the perspective of the individual agent, fairness can be connected directly to the principle of reciprocity~\cite{fehr_fairness_2000, fehr_fairness_2000-1}, where agents act in a certain way because they assume that their action is reciprocated by other agents. This assumption is often a logical prerequisite for envisaging collective action that exists without direct monetary reward~\cite{kahan_logic_2003}. If the system is not spontaneous (i.e. it is designed to steer agents towards a certain goal), the fairness of the multi-agent system can also be assessed from the system designer's perspective~\cite{hunt_origin_2007}. 

In several systems, such as routing~\cite{jiang_graph_2020}, traffic control~\cite{kacprzyk_traffic_2010}, or cloud computing~\cite{minarolli_virtual_2013}, fairness of resource allocation is an essential element for ensuring that system function is not interrupted, leaving particular agents depleted or without access to a service. Fairness in this use-case scenario is often associated with the concept of \textit{load-balancing}~\cite{bejerano_fairness_nodate, kleinberg_fairness_2001}. 

In recent years, multiple works have studied fairness as either an emergent property or the optimization criterion in multi-agent systems~\cite{zhang_fairness_2014, jiang_learning_2019, kumar_decaf_2025}, with several of them focusing on proportional fairness, envy-freeness~\cite{jiang_learning_2019}, maxim fairness~\cite{zhang_fairness_2014} and proportionality rewards~\cite{kumar_decaf_2025} as some prominent examples. The Nash social welfare notion of fairness has also been studied in the context of multi-agent sequential decision making in centralized~\cite{Hossain2021} and decentralized~\cite{Wang2026} settings, showing that it can limit information aggregation and slow down convergence. In contrast, this work separate efficiency and fairness, while also supports the optimization of agents' comfort. This three-dimensional optimization provides more flexibility for system operators, as well as for agents to adapt to different scenarios and explore appropriate trade-offs.

In multi-agent reinforcement learning problems, the legicographic max-min, generalized Gini social welfare function and proportional fairness have been studied in an architecture with two neural sub-networks that co-optimize efficiency and equity~\cite{Zimmer2021}. The reciprocal of the coefficient of variation has been studied in the transport domain to optimize for equal distances traveled~\cite{Aloor2024}, while inter-agent communication is introduced to optimize for temporal fairness using the expected scalarized return (ESR) rather than scalarized expected return (SER) to ensure fairness within and across each time step~\cite{Chouaki2026}. These approaches are limited to a very low number of agents, in contrast to this work that provides a model and its applicability to scale optimization to thousands of agents.

\section{System Optimization Modeling}\label{Section:Methodology}
This section illustrates the modeling of the problem and the formulation of the agents' behavior. Subsection \ref{Methodology:Modelling} shows the general overview of the model, while Section \ref{Methodology:Formulation} introduces the precise formalization of the model.

\subsection{Scenario Modeling}\label{Methodology:Modelling}

This work assumes multiple agents collectively planning their actions, taking into account the impact of those actions on system-wide efficiency (global cost) and the level of their individual discomfort (local cost) they experience. In such a model, each agent self-determines a set of options (plans) from which it can choose, representing the operational flexibility of the agent. The agent's choice is communicated to other agents in the network. All agents optimize their respective choices, taking into account (i) system-wide efficiency (impact of the plan selection on the system performance) and (ii) the level of their individual discomfort caused by the selection (how much comfort they are sacrificing by this selection). 

This model is relevant for various scenarios in collaborative planning, including Smart Grid scenarios, in which the goal is to schedule household appliances over a day in accordance with citizens' preferences, with the aim of reducing the likelihood of a blackout\footnote{We assume blackouts when energy peaks exceed a certain threshold.}~\cite{Fanitabasi2020}. Each household (agent) has a set of options: \textit{energy profiles} or \textit{schedules}, that represent energy demand throughout the day and correspond to the energy consumption of appliances scheduled for different time points or consumption levels~\cite{Fanitabasi2020}. The schedule (energy profile) is defined as energy consumption over time. The profiles can be ranked by their convenience (i.e. thermal/temperature), a measure that expresses the comfort of using each schedule. The sum of all selected energy demand profiles in the grid should not exceed the blackout threshold, i.e., the total energy consumption in the grid should be stabilized over time. Similarly, the same framework can be used to model a bike-sharing scenario over time, in which an agent represents a citizen. The goal is to reduce the number of manual bike relocations, which are costly and time-consuming, by regulating where citizens pick up and return bikes, while accounting for the total relocation distance. In both cases, the agent selection is guided not only by the discomfort cost but also by the choices of other agents that co-determine system performance.

\subsection{Problem Formulation}\label{Methodology:Formulation}

Table~\ref{tab:notation} introduces the mathematical notations used throughout this paper. Assume a network of $N$ agents embedded in an arbitrary graph structure. Let $a$ denote an agent and $\mathcal{A}$, $(|\mathcal{A}| = N)$, a set of all agents in the network. An agent $a \in \mathcal{A}$ self-determines a set of possible plans $\mathcal{P}_a$ which represents a set of options an agent can choose from. Multiple plans are the flexibility that each agent contributes to the system to optimize. A possible plan $\mathbf{p} \in \mathcal{P}_a$ is a vector of length $d$ $(\mathbf{p} \in \mathbb{R}^d)$ and represents one option, e.g. a schedule for appliances in a Smart Grid scenario\footnote{Since schedules may represent energy consumption over time, different plans represent a demand shift~\cite{Fanitabasi2020}.}, or a trip in a bike-sharing scenario~\cite{Pournaras2020}. All agents have possible plans of the same length $d$. Each possible plan is associated with a \textit{cost}. Let $l(\mathbf{p})$ be a \textit{discomfort cost function} of the plan $\mathbf{p}$, defined as follows: $l:\mathbb{R}^d \xrightarrow{} \mathbb{R}$ and its value when applied to plan $\mathbf{p}$ is referred to as the \textit{cost} of the plan $\mathbf{p}$. The discomfort cost may be expressed as the distance to the preferred biking station. It can represent the price or the required resources to execute a certain plan.

\begin{table}[!htb]
	\centering
	\caption{Table of notation used throughout this work.}
	\resizebox{\columnwidth}{!}{
		\begin{tabular}{ll}
			\toprule
			\textbf{Symbol} &  \textbf{Definition} \\
			\midrule
			$N$ & Number of agents \\
			$T$ & Number of iterations \\
			$\mathcal{A}$ & Set of agents \\
			$\mathcal{P}_a $ & Set of possible plans for agent $a$ \\
			$\mathcal{T}_a$ & Set of children of agent $a$ \\
			$a, r \in \mathcal{A}$ & Agent $a$, root agent $r$\\
			$p_a$, $c \in \mathcal{T}_a$ & Parent and child of agent $a$ \\
			$\mathbf{p} \in \mathcal{P}_a$ & Plan from a set of possible plans for agent $a$ \\
			$s_a^{(t)} \in \mathcal{P}_a$ & Plan selected by agent $a$ at iteration $t \in T$ \\
			$\tilde{s}_a^{(t)} \in \mathcal{P}_a$ & Preliminary plan selected by agent $a$ at iteration $t \in T$ \\
			$\mathbf{g}^{(t)}$ & Global response (sum of all plans) \\
			$g$, $l$ & Inefficiency and discomfort cost function ($\mathbb{R}^d \rightarrow \mathbb{R}$) \\
			$G$, $L$, $U$ & Inefficiency, average discomfort and unfairness cost \\
			$\mathbf{a}_a^{(t)}$ & Aggregated response of agent $a$ at time $t$ \\
			$\tilde{\mathbf{a}}_a^{(t)}$ & Preliminary aggregated response of agent $a$ at time $t$ \\
			$k_a^{(t)}$ & Aggregated plan cost of agent $a$ at time $t$ \\
			$\tilde{k}_a^{(t)}$ & Preliminary aggregated plan cost of agent $a$ at time $t$ \\
			$K^{(t)}_{a}$ & Aggregated squared plan cost of agent $a$ at time $t$ \\
			$\tilde{K}^{(t)}_{a}$ & Preliminary aggregated squared plan cost of agent $a$ at time $t$ \\
			$\gamma, \beta, \alpha \in \mathbb{R}$ & Weights for inefficiency, discomfort and unfairness cost \\
			$i_{a}^{(t)} \in \{0, 1\}^{|\mathcal{T}_a|}$ & Binary acceptance vector of client $a$ at round $t$ \\
			$\hat{i}_{a}^{(t)} \in \{0, 1\}^{|\mathcal{T}_a|}$ & Preliminary binary acceptance vector of client $a$ at round $t$ \\
			\bottomrule
	\end{tabular}}\label{tab:notation}
\end{table}

We further assume that agents reassess their selected plans iteratively. At each iteration $t \in T$, each agent $a$ selects one of the possible plans, referred to as the \textit{selected plan} and noted as $s_a^{(t)} \in \mathcal{P}_a$. The sum of selected plans of all agents in the network is referred to as a \textit{global response}, $\mathbf{g}^{(t)} \in \mathbb{R}^d$:

\begin{equation}
    \mathbf{g}^{(t)} = \sum_{a \in \mathcal{A}}s_a^{(t)}.
\end{equation}

Intuitively, in the energy domain, the global response $\mathbf{g}^{(t)}$ may represent the total power consumption over 24h in the Smart Grid. Alternatively, in the bike sharing scenario, it may represent all the incoming and outgoing bikes in a bike sharing station. The global response is constructed via cooperation of agents with respect to one or more objectives, i.e. depending on objective(s), a global response with different characteristics can be constructed, by self-regulating selections at each agent.

The fitness of the agent selections is evaluated according to one or more objectives. The efficiency objective is abstracted by a cost function $g: \mathbb{R}^d \xrightarrow{} \mathbb{R}$ that takes the global response $\mathbf{g}^{(t)}$ as input. The value of the function $g$ applied to global response at time $t$, $g^{(t)}$, is referred to as the \textit{inefficiency cost}:

\begin{equation}
    G^{(t)} = g(\mathbf{g}^{(t)}) = g(\sum_{a \in \mathcal{A}}s_a^{(t)}).
    \label{eq:global_cost}
\end{equation}

\noindent The discomfort objective is assessed by the \textit{average discomfort cost}, $L^{(t)} \in \mathbb{R}$, defined as the average over individual discomfort costs of the agents' selected plans, or more formally:

\begin{equation}
    L^{(t)} = \frac{1}{N}\sum_{a \in \mathcal{A}}l(s_a^{(t)}).
	\label{eq:local_cost}
\end{equation} 

We assume that agents can communicate with each other to adapt their plan choices and collaboratively optimize the objectives. Communication and coordination are necessary, especially when the inefficiency and discomfort cost functions are nonlinear and minimizing independently these costs functions does not necessarily yield a system-wide minimum~\cite{pournaras2018}.

An agent $a$ summarizes the knowledge from all other agents with which it exchanges messages (denoted as $\mathcal{T}_a \subset \mathcal{A}$), including itself, via \textit{aggregation}. Given a possible plan and its discomfort cost, the aggregated knowledge is constructed by computing the following: (i) \textit{aggregated response}, the sum of all received plans and (ii) \textit{aggregated plan cost}, used by the discomfort cost function. The aggregated response of an agent $a$ at time $t, a_{a}^t \in \mathbb{R}^d$, is a sum of selected plans of agent $a$ and all the agents that is connected with:

\begin{equation}	
    \mathbf{a}_a^{(t)} = \sum_{c \in \mathcal{T}_a}s_c^{(t)},
	\label{eq:aggregated-response-definition}
\end{equation}

\noindent where $c \in \mathcal{T}_a$ is an agent linked with agent $a$. The aggregated plan cost at time $t$ is a sum of discomfort costs of the selected plans of agent $a$ and all the agents that is connected with:

\begin{equation}
    k_a^{(t)} = \sum_{c \in \mathcal{T}_a}l(s_c^{(t)}).
	\label{eq:agg-plan-cost}
\end{equation}

\section{Modeling Fairness Optimization}\label{Section:Decentralised_Optimisation}

This section introduces the optimization model of fairness and how it can complement a multi-objective optimization of efficiency and comfort. 

\subsection{Unfairness Minimization}\label{Framework:Unfairness_Minimisation}

The model illustrated above takes into consideration the inefficiency cost (Equation \ref{eq:global_cost}) and average discomfort cost (Equation \ref{eq:local_cost}). However, it does not account for the distribution of discomfort costs across the agents. It is possible that although the inefficiency and average discomfort costs are minimized, the distribution is heavily skewed, causing to a selected subset of agents much higher discomfort levels than the rest of the set. Formally, the unfairness cost, $U^{(t)} \in \mathbb{R}$ is defined as the standard deviation between the costs of the selected plans, or more formally:

\begin{equation}
    U^{(t)}  = \sqrt{\frac{1}{N}\sum_{c \in \mathcal{A}}(l(s_c^{(t)}) - L^{(t)})^2}.
	\label{eq:unfairness-definition}
\end{equation}

\noindent The dispersion and equality of discomfort costs can also be calculated with other metrics~\cite{Xinying2023}. For instance, the coefficient of variation is scale invariant, while the Jain's fairness index and the Gini coefficient are both scale invariant and bounded. We show here how the Equation 6 can be calculated and optimized in decentralized settings and the principle can be extended to other measures as well (see Appendix~\ref{Appendix:Fairness}). 

Minimizing unfairness (Equation~\ref{eq:unfairness-definition}) requires knowledge of the discomfort costs of all agents in $\mathcal{A}$, which is not available a priori to each individual agent in a decentralized network. However, we show that Equation~\ref{eq:unfairness-definition} can be rewritten using only $k_a^{(t)}$ and $K_a^{(t)}$ (aggregated and aggregated squared plan costs). Since each agent $a$ computes those only with respect to other agents with which it is linked as $\mathcal{T}_a \subseteq \mathcal{A}$, this yields an approximation of system-wide unfairness. Since agents can adjust their plan selections to minimize costs, this approximation can also iteratively improve. Moreover, since the aggregated plan cost $k_a^{(t)}$ is locally available to agents to calculate the average discomfort cost, no additional communication overhead is required for the fairness calculation. We introduce though a low-cost local computation of the \textit{aggregated squared plan cost} at time $t$, $K_{a}^{(t)} \in \mathbb{R}$, that represents the sum of \textit{squared} discomfort costs of the selected plans of agent $a$ and all the agents that it is connected with:

\begin{equation}
    K^{(t)}_{a} = \sum_{c \in \mathcal{T}_a}(l(s_c^{(t)}))^2.
	\label{eq:agg-sq-plan-cost}
\end{equation}

We show below that it is possible to rewrite Equation~\ref{eq:unfairness-definition} using only the quantities of aggregated plan cost and aggregated squared plan cost, which are in direct disposition of each agent. Moreover, once rewritten in that form, the Equation~\ref{eq:unfairness-definition} is a square root of the (biased) sample variance, and it converges to the true sample variance as the sample size increases. To this end, we introduce the following lemma and its corollary:

\begin{lemma}

The aggregated plan cost $k_a^{(t)}$, aggregated squared plan cost $K_a^{(t)}$ and the number of agents $N$ are sufficient for agents connected in a tree network to compute unfairness as: 

\[
U^{(t)} = \sqrt{\frac{1}{N}K_a^{(t)} - \left(\frac{1}{N}k_a^{(t)}\right)^2},
\]

\noindent where $k_a^{(t)}$ and $K_a^{(t)}$ are defined in Equations~\ref{eq:agg-plan-cost} and~\ref{eq:agg-sq-plan-cost} and evaluated over $\mathcal{T}_a = \mathcal{A}$ for $N=|\mathcal{A}|$.
\end{lemma}

\begin{proof}
Expanding on the square of Equation~\ref{eq:unfairness-definition}, we get: 
\begin{align*}
(U^{(t)})^2
&= \frac{1}{N}\sum_{c \in \mathcal{A}}\bigl(l(s_c^{(t)}) - L^{(t)}\bigr)^2 \\
&= \frac{1}{N}\sum_{c \in \mathcal{A}} l(s_c^{(t)})^2
   - 2L^{(t)}\cdot\frac{1}{N}\sum_{c \in \mathcal{A}} l(s_c^{(t)})
   + \bigl(L^{(t)}\bigr)^2 \\
&= \frac{1}{N}K_a^{(t)}
   - 2L^{(t)}\cdot\frac{1}{N}\sum_{c \in \mathcal{A}} l(s_c^{(t)})
   + \bigl(L^{(t)}\bigr)^2 \\
&= \frac{1}{N}K_a^{(t)} - 2L^{(t)}\cdot\frac{1}{N}k_a^{(t)} + \bigl(L^{(t)}\bigr)^2 \\
&= \frac{1}{N}K_a^{(t)} - 2\cdot\frac{1}{N}k_a^{(t)}\cdot\frac{1}{N}k_a^{(t)} + \left(\frac{1}{N}k_a^{(t)}\right)^2 \\
&= \frac{1}{N}K_a^{(t)} - 2\left(\frac{1}{N}k_a^{(t)}\right)^2 + \left(\frac{1}{N}k_a^{(t)}\right)^2 \\
&= \frac{1}{N}K_a^{(t)} - \left(\frac{1}{N}k_a^{(t)}\right)^2.
\end{align*}
Since $U^{(t)} \geq 0$, taking the square root of both sides gives
\[
U^{(t)} = \sqrt{\frac{1}{N}K_a^{(t)} - \left(\frac{1}{N}k_a^{(t)}\right)^2}. \qedhere
\]
\end{proof}

\begin{corollary}
Let $X$ be a real-valued random variable with $\mathbb{E}[X^2] < \infty$, and let the plan costs $X_c = l(s_c^{(t)})$, $c \in \mathcal{A}$, $|\mathcal{A}| = N$, be independent and identically distributed samples from the distribution of $X$. Then $(U^{(t)})^2$ is the (biased) sample variance of $X_1,\dots,X_N$, and as $N \to \infty$, it almost surely holds that: 
\[
(U^{(t)})^2 \xrightarrow{a.s.} \mathrm{Var}(X), \qquad
U^{(t)} \xrightarrow{a.s.} \sqrt{\mathrm{Var}(X)}.
\]
\end{corollary}

\begin{proof}
By Lemma~1, $(U^{(t)})^2 = \frac{1}{N}\sum_{c \in \mathcal{T}_a} X_c^2 - \left(\frac{1}{N}\sum_{c \in \mathcal{T}_a} X_c\right)^2$, which is the (biased) sample variance of $X_1,\dots,X_N$. Since $\mathbb{E}[X^2] < \infty$, the strong law of large numbers gives:
\[
\frac{1}{N}\sum_{c \in \mathcal{T}_a} X_c \xrightarrow{a.s.} \mathbb{E}[X],
\quad
\frac{1}{N}\sum_{c \in \mathcal{T}_a} X_c^2 \xrightarrow{a.s.} \mathbb{E}[X^2].
\]
Therefore:
\begin{align*}
(U^{(t)})^2
&\xrightarrow{a.s.} \mathbb{E}[X^2] - (\mathbb{E}[X])^2 \\
&= \mathbb{E}\big[(X - \mathbb{E}[X])^2\big] \\
&= \mathrm{Var}(X).
\end{align*}
Taking square roots yields: $U^{(t)} \xrightarrow{a.s.} \sqrt{\mathrm{Var}(X)}$.
\end{proof}

Given the above results, it becomes clear that the variance of the incurred discomfort costs (i) can be directly computed (approximated) from the information already in possession of each agent, without any additional messages, and (ii) converges to the true variance as $N$ increases. Furthermore, the same approach using only the aggregated plan cost and aggregated squared plan cost can be used to derive more complex fairness measures. In Appendix \ref{Appendix:Fairness} we illustrate one additional measure of fairness, namely the Jain's fairness index~\cite{sediq_optimal_2012}. It should be noted that Corollary~1 relies on the assumption of independent and identically distributed samples, which may not be strictly satisfied in the decentralized setting where agents make coupled decisions to minimize a shared objective. In practice though, the derived formula serves as a tractable proxy for system-wide unfairness, enabling its estimation without any additional communication overhead.

\subsection{Multi-objective Optimization}
\label{Framework:Collaborative_Optimisation}

The model combines three (orthogonal) objectives: minimizing inefficiency cost (Equation \ref{eq:global_cost}), discomfort cost (Equation \ref{eq:local_cost}) and unfairness cost (Equation \ref{eq:unfairness-definition}). The objectives are scalarized, i.e. costs are linearly combined into a single additive complex cost. The optimization problem to solve is then formulated $\forall a \in \mathcal{A}$ as follows:

 \begin{equation}
  \label{eq:objective}
  \begin{aligned}
  s_a^{(t)}=\argminimal{\substack{\mathbf{p}_a \in \mathcal{P}_a }} \quad
      & \gamma G^{(t)} + \alpha U^{(t)} + \beta L^{(t)}, \\
      G^{(t)}
      & = g\left(\sum_{a \in \mathcal{A}} \mathbf{p}_a \right), \\
      L^{(t)}
      & = \frac{1}{N}\sum_{a \in \mathcal{A}} l(\mathbf{p}_a), \\
      U^{(t)}
      & = \sqrt{\frac{1}{N}\sum_{a \in \mathcal{A}}\bigl(l(\mathbf{p}_a) - L^{(t)}\bigr)^2},
  \end{aligned}
  \end{equation}
  
\noindent where $s_a^{(t)} \in \mathcal{P}_a$ is the plan selected by agent $a \in \mathcal{A}$, and $\alpha, \beta, \gamma \in \mathbb{R}$ are weights such that $0 \leq \alpha, \beta, \gamma \leq 1$ and $\alpha + \beta + \gamma = 1$. The weights indicate preference towards a certain goal. They can be uniformly selected for the whole system or each agent can set its own preference towards certain goals. In this case, $\alpha_a, \beta_a, \gamma_a$ are a triplet of weights that expresses the preference of agent $a$ towards certain optimization goals. Agents are assumed to have intrinsic or extrinsic incentives (e.g. monetary rewards) to be altruistic, i.e. high $\gamma$, $\alpha$ values and low $\beta$ value. For instance, this may include environmental friend energy consumers prioritizing low thermal comfort and high responsiveness to the availability of renewable energy resources, or responsiveness to low energy prices.

Since the objectives are scalarized and plans are discrete, the problem becomes a standard multi-objective combinatorial problem, where we search for a plan $s_a \in \mathcal{P}_a$ for each agent $a \in \mathcal{A}$ that minimizes Equation~\ref{eq:objective}. The formulation is solver-agnostic and can be integrated to any combinatorial optimization method. To assess the model applicability, we apply the model to I-EPOS~\cite{pournaras2018,Pournaras2020}, which supports decentralized combinatorial optimization in a collaborative manner.

\section{Model Applicability to Collective Learning}\label{Section:I-EPOS}

The I-EPOS~\cite{pournaras2018} algorithm has been supporting decentralized combinatorial optimization of two or more objectives. As such we select it as a basis of assessing the applicability of the proposed model, extending it to optimize for fairness as illustrated in Section~\ref{Framework:Unfairness_Minimisation}. Another alternative algorithm for applicability is COHDA~\cite{hinrichs2013}. We focus here on I-EPOS, mostly because in COHDA, agents exchange their full knowledge, unlike I-EPOS agents that iteratively adjust their plan selections based on the aggregated selection of others members in the network. One advantage of COHDA is that it accounts for arbitrary graph structures, whereas I-EPOS relies on agents that (self-)organize into hierarchical tree topologies for efficient information aggregation and incremental decision making~\cite{pournaras2010}. 

Organizing agents in a tree topology can be achieved using multiple published algorithms, including AETOS~\cite{pournaras2010}, and thus constitutes an acceptable relaxation of the assumptions presented in Section~\ref{Section:Methodology}. Given the tree topology, one of the agents $a \in \mathcal{A}$ occupies the position of a root agent $r$, while the rest $|\mathcal{A}| -1$ agents are placed on a lower level of the tree\footnote{In the subsequent parts of this article, we assume a balanced tree topology with the same branching factor for all nodes of the tree.}. The agents (i) send messages to a set of child agents and (ii) receive messages from the parent agent (denoted as $p_a$). The I-EPOS algorithm iteratively optimizes the plan selections, firstly by propagating preliminary plans from the bottom of the tree up to the root agent (\textit{the bottom-up phase}) and then by propagating knowledge about accepted (or rejected) plans from the root agent down to the bottom of the tree (\textit{the top-down phase}). Both of those phases are outlined in Section~\ref{I-EPOS:Bottom-up} and~\ref{I-EPOS:Top-Dowm}.

\subsection{Bottom-up Phase}
\label{I-EPOS:Bottom-up}

In the bottom-up phase, agents make \textit{preliminary} plan choices that may be accepted based on the received feedback from the agents placed above them in the tree, or may revert to the earlier selection. Formally, every iteration $t \in T$ starts from the agent receiving a preliminary aggregated response from its children, i.e. $\tilde{a}^{(t)}_a = \sum_{c \in \mathcal{T}_a}\tilde{s}^{(t)}_c$. Those aggregated plans are used to calculate all preliminary values of the cost functions, including the preliminary aggregated plan cost $\tilde{k}^{(t)}_a = \sum_{c \in \mathcal{T}_a }l(\tilde{s}^{(t)}_c)$ and preliminary aggregated squared plan cost $\tilde{K}_a^{(t)} = \sum_{c \in \mathcal{T}_a} (l(\tilde{s}^{(t)}_c))^2$. Since the plan is preliminary, it is communicated to the agent's parent, but the agent may revert to previously selected plans if this reduces costs, jointly with other agent selections. The process begins from the bottom of the tree, hence referred to as the \textit{bottom-up phase}.

Formally, since the topology of the tree remains stable during the optimization, each agent is able to track the progress by keeping the previous\footnote{For the first iteration $t=0$, no previous plan information is considered.} values of its own selected plan $s_a^{(t-1)}$ and aggregated responses of all its children $a_c^{(t-1)}$. Let $\tilde{\boldsymbol{i}}_a^{(t)}$ denote the binary acceptance vector of agent $a$ at iteration $i$ and let $\tilde{\boldsymbol{i}}_{a_c}^{(t)}$ denote the element of this vector that corresponds to a child $c \in \mathcal{T}_a$. During the preliminary bottom-up phase, each agent performs a limited-size brute-force search over a discrete space $I = \{0, 1\}^{|\mathcal{T}_a|}$, selecting such $\tilde{\boldsymbol{i}}_a^{*(t)}$ that minimizes\footnote{At the first iteration $t=0$, all proposed changes are accepted.} the scalarized objective as shown below: 

  \begin{equation}
  \begin{aligned}
  \boldsymbol{i}^{*}
      &= \operatorname*{arg\,min}_{\boldsymbol{i} \in I} \quad
         \gamma_a \cdot g(\boldsymbol{g}^{t-1} + \Phi_{a_c}) \\
      &\quad + \alpha_a \cdot
         \sqrt{\frac{1}{N}(K_r^{(t-1)} + \Lambda_{a_c})
         - \frac{1}{N^2}(k_r^{(t-1)} + \Gamma_{a_c})^2} \\
      &\quad + \beta_a \cdot \frac{1}{N}(k_r^{(t-1)} + \Gamma_{a_c}), \\
  \Phi_{a_c}
      &= \sum_{c \in \mathcal{T}_a} \boldsymbol{i}_{a_c} \cdot \Delta \tilde{a}_{c}^{(t)}, \\
  \Gamma_{a_c}
      &= \sum_{c \in \mathcal{T}_a} \boldsymbol{i}_{a_c} \cdot \Delta \tilde{k}^{(t)}_{c}, \\
  \Lambda_{a_c}
      &= \sum_{c \in \mathcal{T}_a} \boldsymbol{i}_{a_c} \cdot \Delta \tilde{K}^{(t)}_{c},
  \end{aligned}
  \end{equation}

\noindent where $\Delta \tilde{\mathbf{a}}_{c}^{(t)} = \tilde{\mathbf{a}}^{(t)}_{c} - \mathbf{a}_{c}^{(t-1)}$ and respectively for $\Delta \tilde{k}_{c}^{(t)}$ and $\Delta \tilde{K}^{(t)}_{c}$. Moreover $\alpha_a, \beta_a, \gamma_a$ are the importance weights corresponding to agent $a$ such that $0 \leq \gamma_a, \alpha_a, \beta_a \leq 1$ with $\gamma_a + \alpha_a + \beta_a = 1$. Each agent accepts the proposed (aggregated) change of its children only when it is improving the scalarized objective. Similarly, the proposed plan $\tilde{s}^{(t)}_a \in P_a$ is selected to minimize the scalarized cost function together with an acceptance vector $\tilde{\boldsymbol{i}_a^{(t)}}$. Finally, the aggregated preliminary response $\tilde{a}^{(t)}_a = \tilde{s}_a^{(t)} + \sum_{c \in \mathcal{T}_a}a_c^{(t-1)} + \sum_{c \in \mathcal{T}_a} \tilde{\boldsymbol{i}}^{(t)}_{a_c} * \Delta \tilde{a}^{(t)}_c$ is sent upwards to the parent. The process iteratively continues until the aggregated responses reach the root agent.

\subsection{Top-Down Phase}
\label{I-EPOS:Top-Dowm}

In the top-down phase, agents determine whether the \textit{preliminary} plan should be accepted and treated as \textit{effective}, or they should revert to the last accepted plan. The decision is based on information received from the parent (acceptance of the parent) and the agent's own decision. The process starts from the root agent $r$ and continues to the bottom of the tree. This phase accounts for the information gap of agents during the bottom-up phase: agents coordinate and optimize plan selections based on information they have for the agents below. However, they do not have information about the agents above. Therefore, the top-down phase is another opportunity of the agents to improve their plan selections by keeping their preliminary selection or switching back to the earlier one. 

A top-down phase begins when preliminary aggregates $\tilde{a}^{(t)}_a$ reach the root agent $r$. The root agent automatically accepts its preliminary acceptance vector and plan, i.e. $\tilde{\boldsymbol{i}}_r^{(t)} = \boldsymbol{i}^{(t)}_r \wedge \tilde{s}_r^{(t)} = s^{(t)}_r$. The effective acceptance vector $\boldsymbol{i}_c^{(t)}$ (together with effective global response, effective aggregated plan cost and effective aggregated squared plan cost) is communicated to children of the root agent $c \in \mathcal{T}_r$. Each agent $a$ computes their effective acceptance vector by taking a Hadamard product of their own preliminary vector and an effective vector communicated by its parent $p_a$, i.e. $\boldsymbol{i}_a^{(t)} = \tilde{\boldsymbol{i}}^{(t)}_a \odot \boldsymbol{i}_{p_a}^{(t)}$. In other words, each agent accepts changes proposed by the child only if: (i) it has accepted those changes (preliminary) during the optimization stage and (ii) those changes are also accepted by its parent. In this way, the knowledge propagates from the root agent $r$ to the bottom of the tree. Both phases of the extended I-EPOS protocol are presented in Figure \ref{fig:I-EPOS}.

\begin{figure*}
    \centering
    \includegraphics[width=0.8\linewidth]{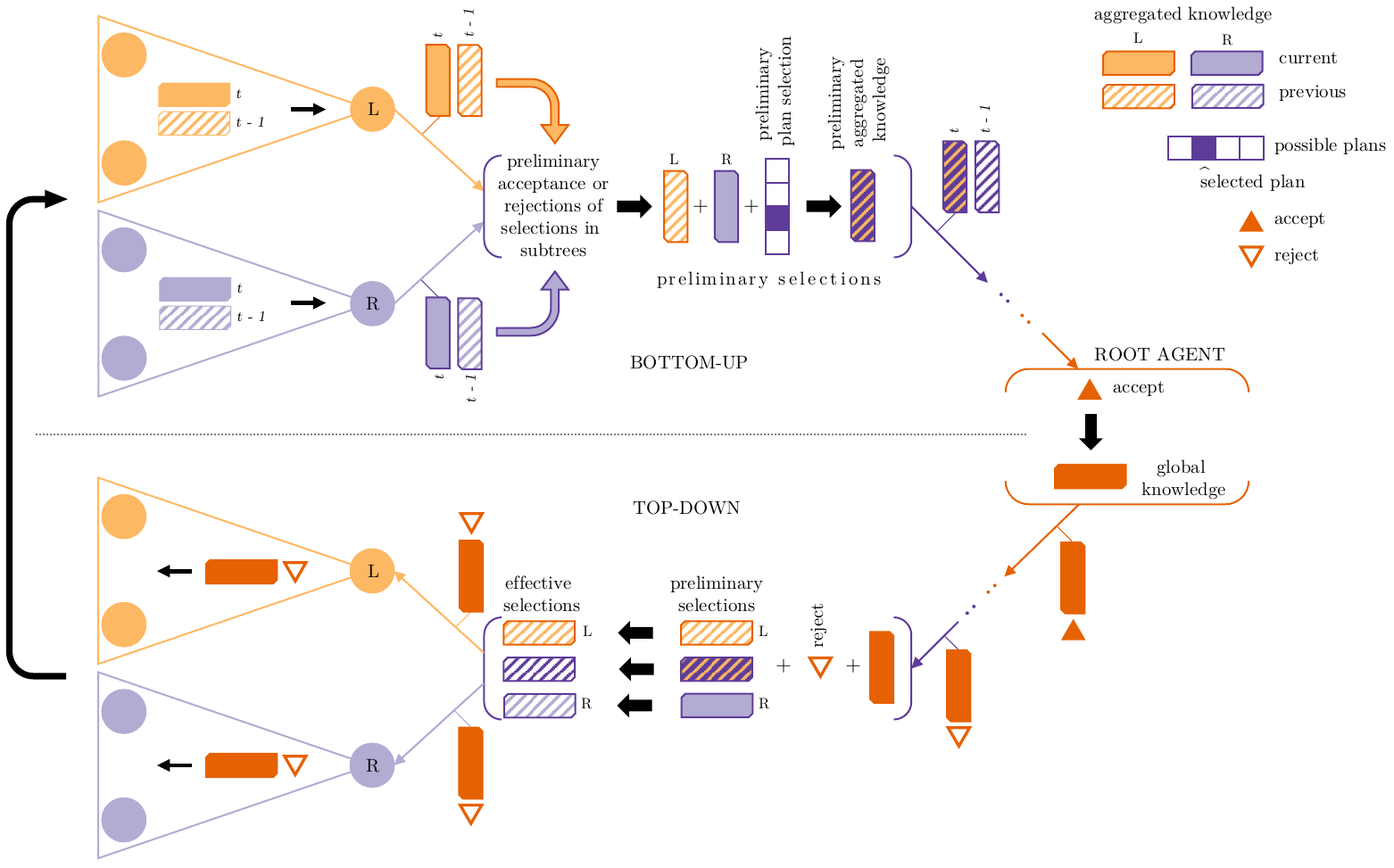}
    \caption{Sequence diagram of the augmented I-EPOS protocol on a balanced tree topology during an iteration. In the bottom-up phase, each agent sends its preliminary plan selection together with the aggregated plan cost $\tilde{k}_a^{(t)}$ and aggregated squared plan cost $\tilde{K}_a^{(t)}$ to its parent to locally optimize the scalarized cost function of the three objectives. In the top-down phase, each agent sends the effective acceptance vectors $\boldsymbol{i}^{(t)}$ alongside the updated global response, aggregated plan cost, and aggregated squared plan cost, allowing each agent to either confirm its preliminary selection or revert to the previous accepted plan.}
    \label{fig:I-EPOS}
\end{figure*}

\section{Experimental Evaluation}\label{Section:Experiments}

The experiments are designed to evaluate three key aspects of the proposed solution: (i) the convergence of the plan selection using three possibly orthogonal objectives (inefficiency and discomfort costs combined with unfairness minimization), (ii) impact of the orthogonal objectives on the frontier of possible solutions found by the agents and (iii) improvement of fairness in discomfort distribution. Section~\ref{Appendx:Datasets} of Appendix outlines the datasets used for conducting the experiments, details of experimental settings and threats to validity. The design of the experiments (including hyperparameter tuning) is illustrated in Section~\ref{Experiments:Design}. Results (together with discussions on them) are shown in Section~\ref{Experiments:Results}.

\subsection{Datasets}\label{Experiments:Datasets}

This work employs three different datasets. The synthetic dataset serves as a baseline for assessing the feasibility of the proposed solution~\cite{pournaras2018}. It consists of $1000$ agent, each having 16 possible plans. The dataset generation is outlined in Appendix~\ref{Appendix:Synthetic}. To further assess the optimality of the solutions found, two additional real-world datasets are used. The Smart Grid Scenario emulates the zonal power transmission in the Pacific Northwest project~\cite{pournaras2017}. It includes $5600$ agents, each with exactly $10$ possible plans. The precise specification of the dataset is found in Appendix~\ref{Appendix:Energy}. The bike-sharing use case relies on data obtained from the Hubway Data Visualization Challenge~\cite{pournaras2018} as shown in Appendix~\ref{Appendix:Bicycle}.

\subsection{Design}\label{Experiments:Design}

The experiments leverage the fact that the scalarized objective gives flexibility to fine-tune the optimization preferences via the weights \globalcostweight, \unfairnessweight and \localcostweight. Recall that a higher weight value indicates a stronger preference towards minimizing it. Formally, each solution $z$ is defined as a 3-dimensional vector, $z = (z_1, z_2, z_3) \in \mathbb{R}^3 $, whose elements correspond to the values of the inefficiency, discomfort and unfairness cost, respectively, as measured at convergence (last iteration of the algorithm), $z = (G^{(T)}, L^{(T)}, U^{(T)})$. The set of all possible solutions $z$ is referred to as the \textit{feasible region} and denoted as $\mathcal{Z}$. Given the possible trade-off between the costs of inefficiency, discomfort and unfairness, we can assess the Pareto frontier of the returned solution set $\mathcal{Z}$. The solutions from the Pareto frontier are particularly relevant since they are the ones that express the trade-off, i.e., they exactly show how improving the value of one objective deteriorates another objective and to what extent. The lowest possible cost values of each objective are in the Pareto frontier.

For the experimental setting, agents of I-EPOS organize into a balanced binary tree. Moreover, we set the efficiency objective as a minimization of load imbalances defined as $g(\mathbf{g^{(t)}}) = \frac{1}{N}\sum_{j=1}^d(g^{(t)}_j - \mu)^2$ with $\mu = \frac{1}{n}\sum_{j=1}^dg_j^{(t)}$.

Since the second objective of the experiments is to explore the trade-offs between three potentially conflicting objectives, the design should also allow for an approximation of the feasible region. The feasible region $\mathcal{Z}$ is approximated by a grid-search over the parameters in the following way: first, the values of parameters $\alpha$ (controlling the unfairness) and $\beta$ (controlling the discomfort cost) are both varied between 0 and 1 with step 0.01, creating 5151 possible combinations of parameters for which $\gamma + \alpha + \beta = 1$ ($\gamma$ is always computed when $\alpha$ and $\beta$ are already determined). One combination of parameters is referred to as a \textit{parameter triplet}, $(\gamma, \alpha, \beta)$. Note that during one I-EPOS run, all agents are initialized with the same parameter triplet, which does not change during runtime. Then, I-EPOS runs with each triplet and fed by the possible plans from three different datasets. The full values of all parameters are shown in Appendix~\ref{Appendix:Experiments}.

\subsection{Results}\label{Experiments:Results}

The convergence of all three objectives in relation to executed iterations is displayed in Figure~\ref{moo:convergence}. Only the $\gamma$ and $\alpha$ values are activated with values of 0.5, i.e. the inefficiency cost, while trying to minimize the unequal distribution of the discomfort costs (unfairness). As presented in the Figure~\ref{moo:convergence}, all three objectives are optimized in the span of no more than 35 iterations, with some datasets (synthetic and energy dataset) reaching the plateau of the efficiency cost function in around 10 to 15 iterations. The unfairness cost function exhibits different behavior as it depends on the distribution of the discomfort costs. It remains stable for the synthetic and bicycle datasets, while it drops sharply for the energy dataset. This behavior can be linked with how unequal the distribution of the discomfort costs is. In the case of the energy dataset, the minimization of unfairness is prominent due to the unequal distribution of the discomfort costs, while for the bicycle dataset, that distribution is uniform from the beginning, hence not leaving much room for further improvement.

\begin{figure*}[!htb]
	\centering	
	\subfloat[Synthetic]{%
		\includegraphics[width=0.305\linewidth]{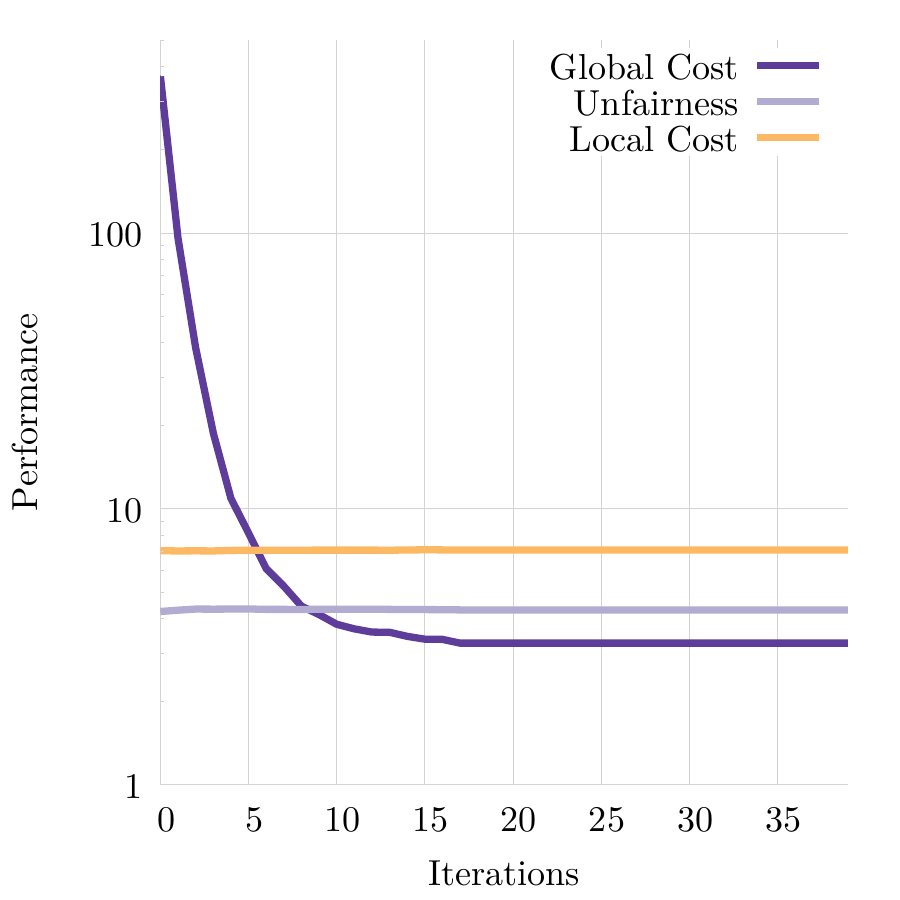}
		\label{moo:convergence-gaussian}}\hfill
	\subfloat[Bicycle]{%
		\includegraphics[width=0.305\linewidth]{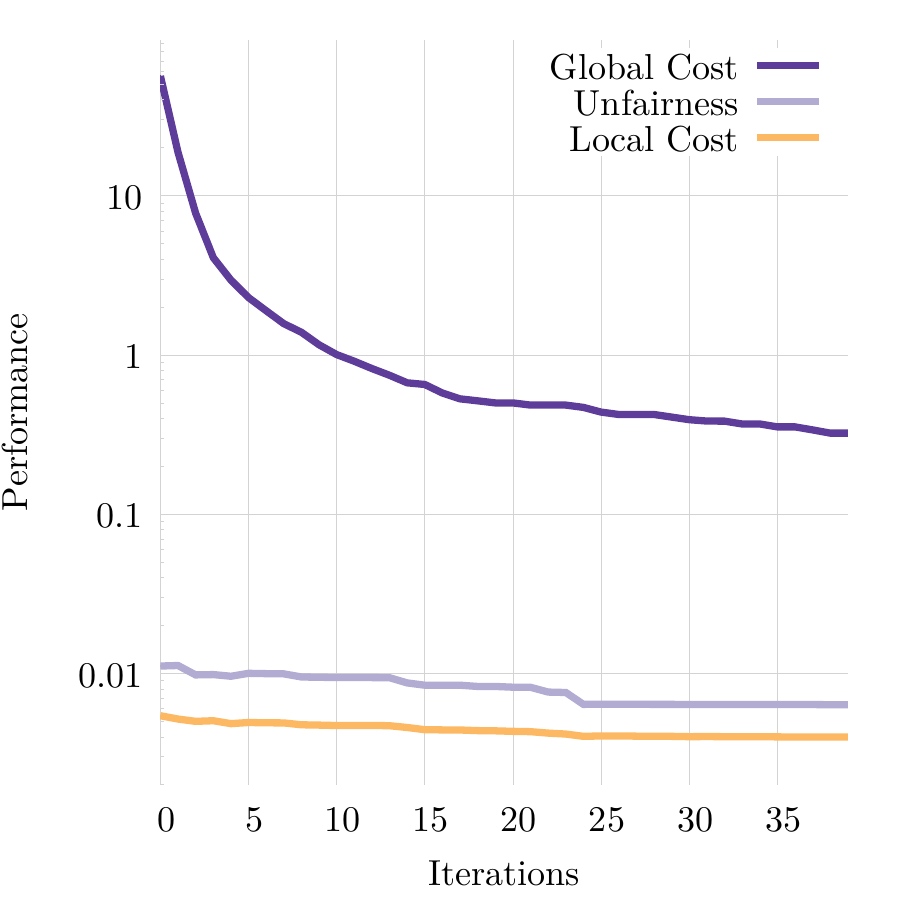}
		\label{moo:convergence-bicycle}}\hfill
	\subfloat[Energy]{%
		\includegraphics[width=0.305\linewidth]{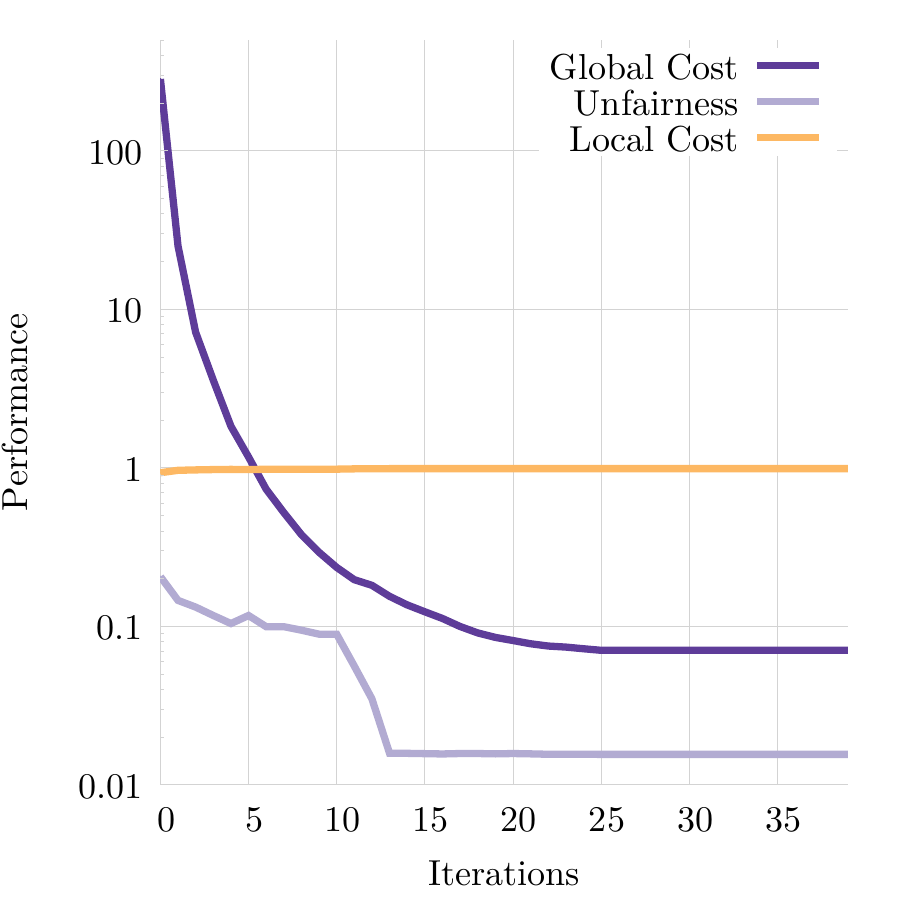}
		\label{moo:convergence-energy}}\hfill
	\caption{Convergence of three objectives over \iepos runtime for $\gamma=0.5$, $\alpha=0.5$, $\beta=0$ and for the (a) synthetic, (b) bicycle and (c) energy datasets. Inefficiency (global), discomfort (local) and unfairness cost reduction represent performance (y-axis).}	
	\label{moo:convergence} 	
\end{figure*}

Figure~\ref{moo:graph-all-datasets} illustrates the discomfort costs of the agents over the hierarchical tree topology, when optimizing for each of the three objectives, i.e. $\globalcostweight=1$, $\localcostweight=1$ or $\unfairnessweight=1$. For clarity, only trees with 100 agents are visualized. We observe that optimizing for efficiency ($\globalcostweight=1$) yields inequalities to how discomfort cost is experienced by different agents (left column). Minimizing the discomfort cost ($\localcostweight=1$) results in equal discomfort costs but inefficiency is ignored and there is no system-wide coordination among agents (middle column). Yet, once we optimize for fairness ($\unfairnessweight=1$), significant efficiency is preserved, while agents show a much higher homogeneous experience of discomfort cost (right column).

\begin{figure}[!htb]
	\centering
	\includegraphics[width=0.4\textwidth]{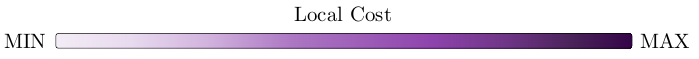}\\
	\textbf{Synthetic}\\
	
	\begin{subfigure}{0.32\columnwidth}
		\includegraphics[width=\linewidth]{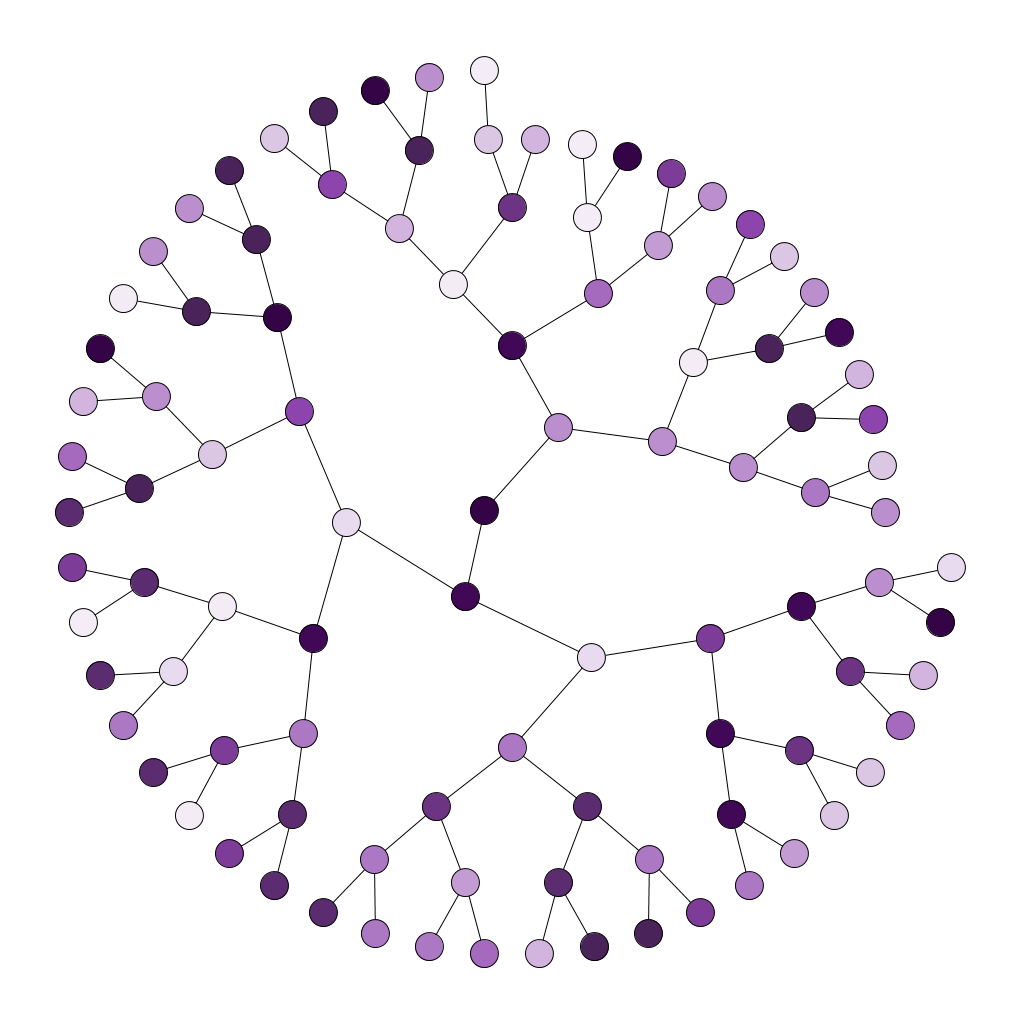}
		\caption{Efficiency ($\globalcostweight=1$)}
		\label{moo:graph-global-cost-gaussian}
	\end{subfigure}
	\hfill
	\begin{subfigure}{0.32\columnwidth}
		\includegraphics[width=\linewidth]{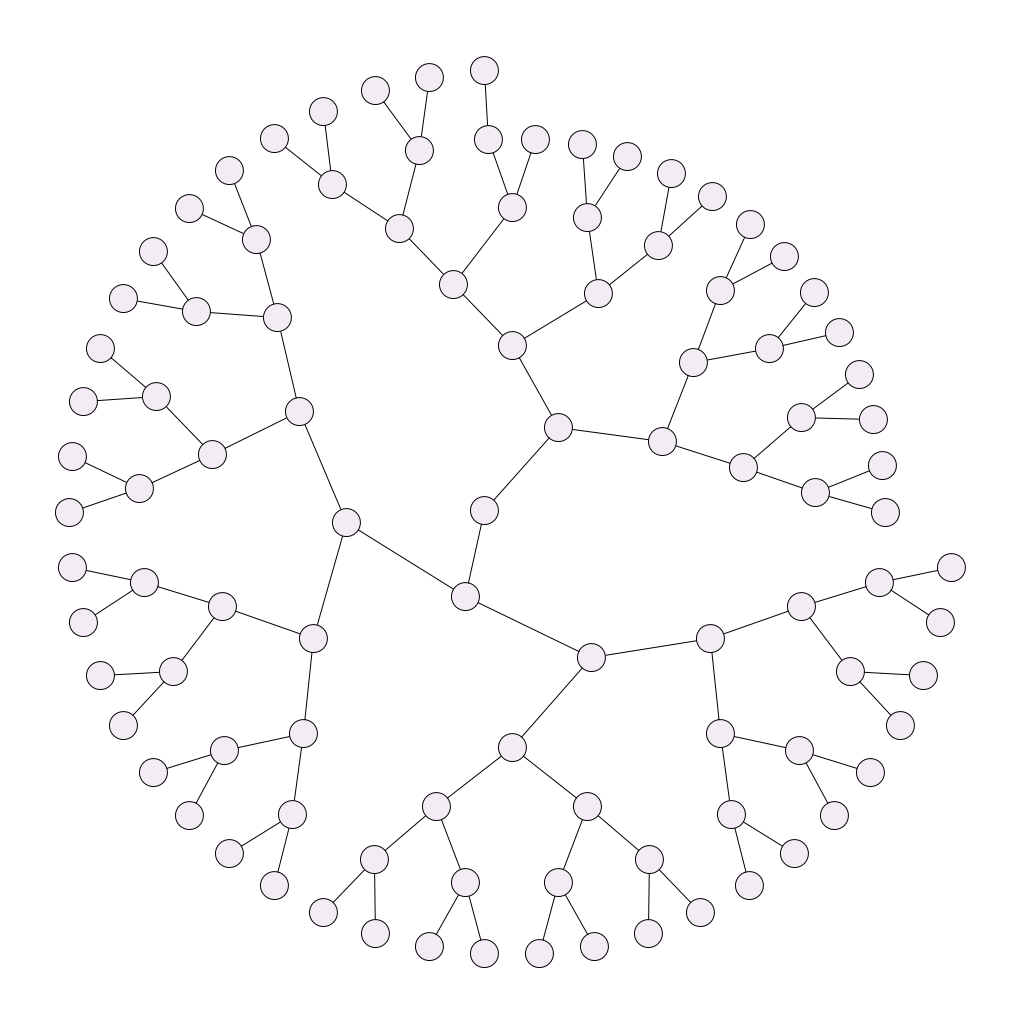}
		\caption{Comfort ($\localcostweight=1$)}
		\label{moo:graph-local-cost-gaussian}
		
	\end{subfigure}
	\hfill
	\begin{subfigure}{0.32\columnwidth}
		\includegraphics[width=\linewidth]{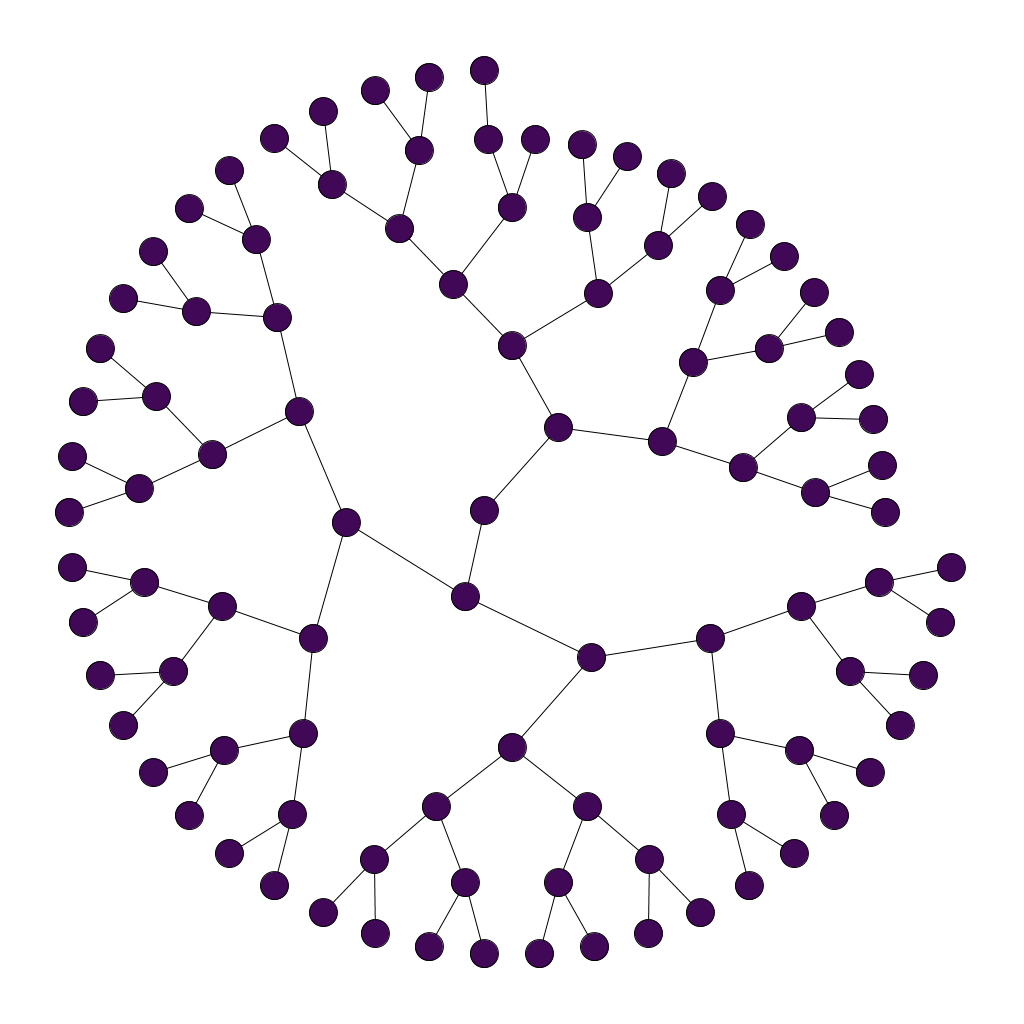}
		\caption{Fairness ($\unfairnessweight=1$)}
		\label{moo:graph-unfairness-gaussian}
	\end{subfigure}\\
	
	\textbf{Bicycle}\\
	
	\begin{subfigure}{0.32\columnwidth}
		\includegraphics[width=\linewidth]{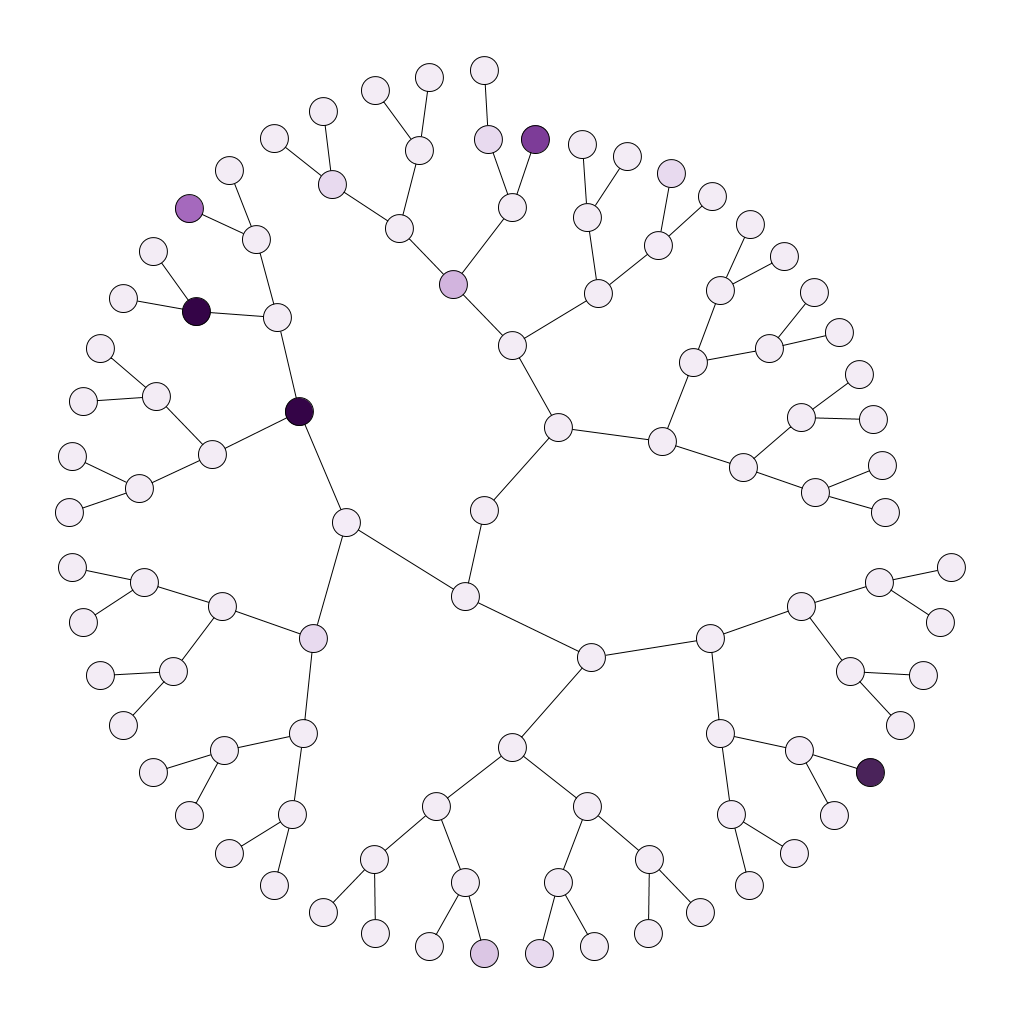}
		\caption{Efficiency ($\globalcostweight=1$)}
		\label{moo:graph-global-cost-bicycle}
	\end{subfigure}
	\hfill
	\begin{subfigure}{0.32\columnwidth}
		\includegraphics[width=\linewidth]{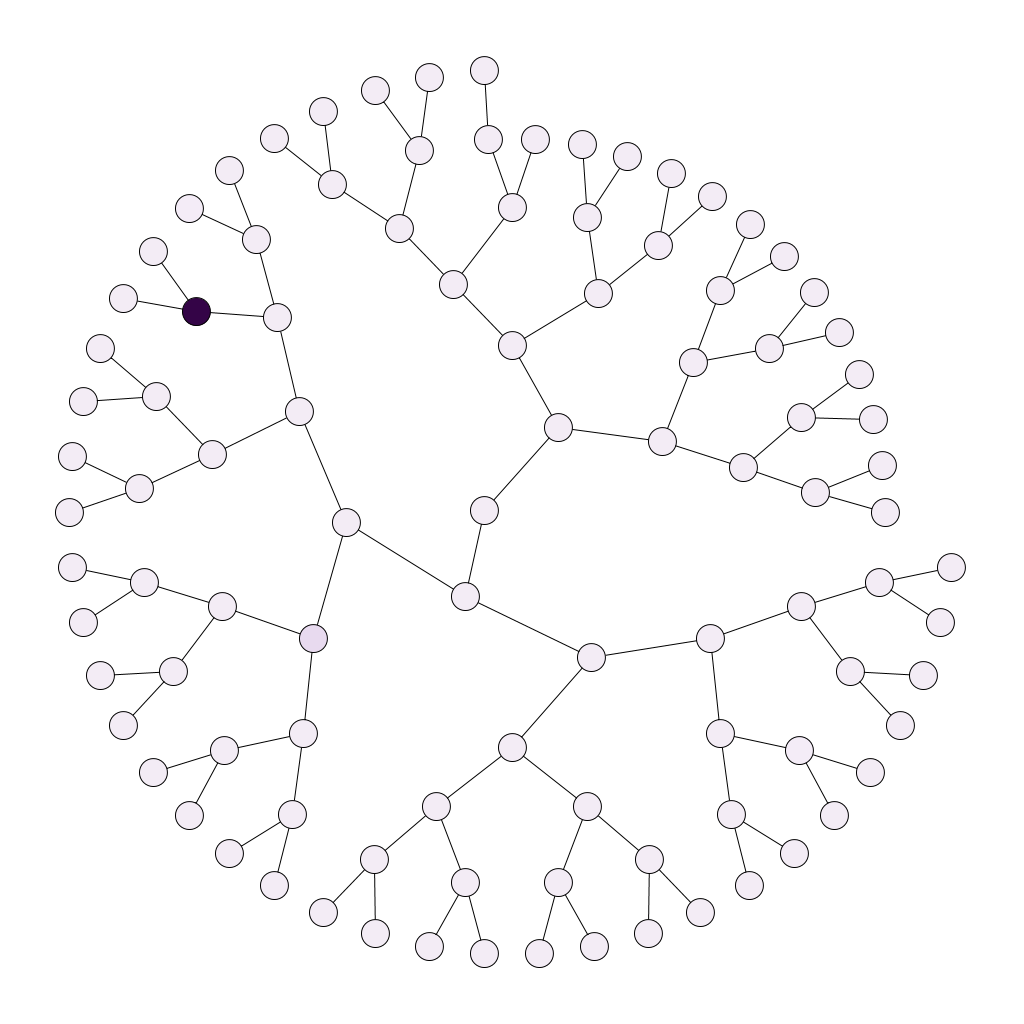}
		\caption{Comfort ($\localcostweight=1$)}
		\label{moo:graph-local-cost-bicycle}
	\end{subfigure}
	\hfill
	\begin{subfigure}{0.32\columnwidth}
		\includegraphics[width=\linewidth]{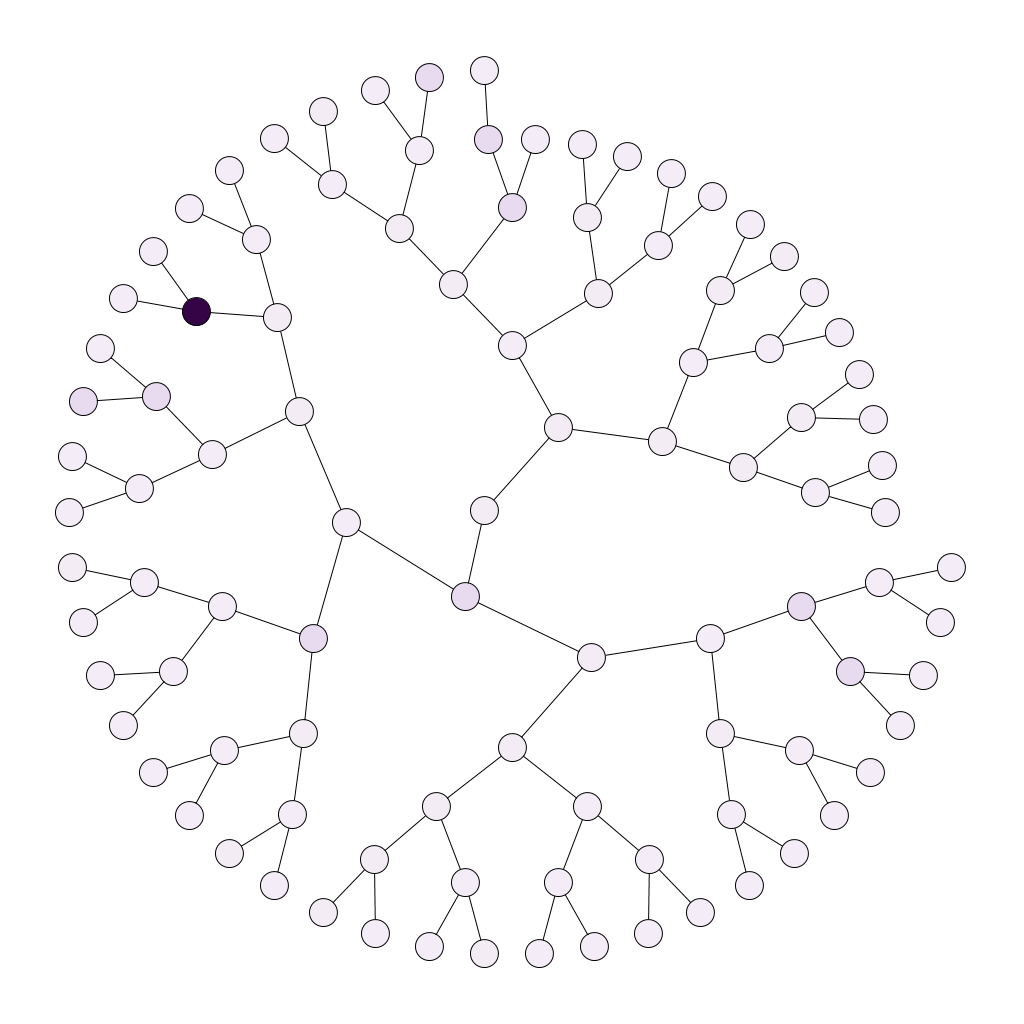}
		\caption{Fairness ($\unfairnessweight=1$)}
		\label{moo:graph-unfairness-bicycle}
	\end{subfigure}
	
	\textbf{Energy}\\
	
	\begin{subfigure}{0.32\columnwidth}
		\includegraphics[width=\linewidth]{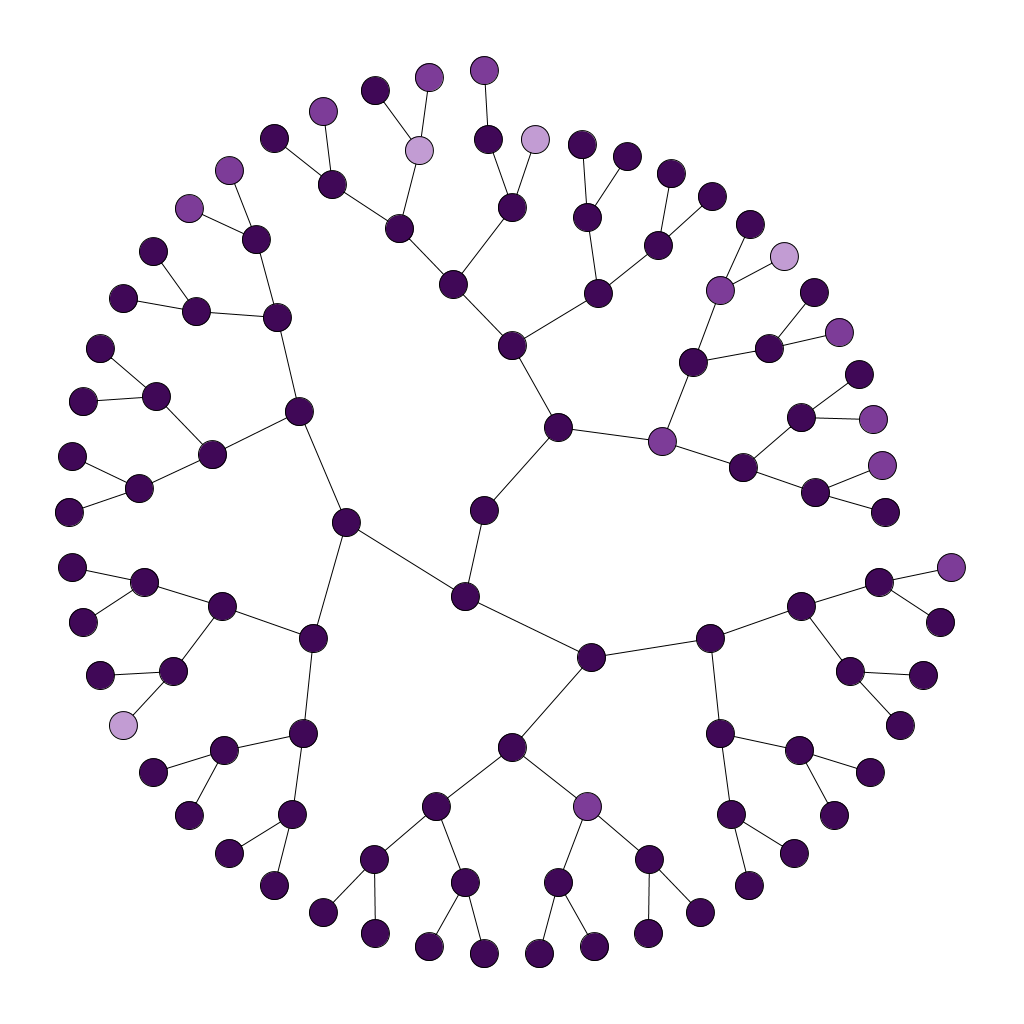}
		\caption{Efficiency ($\globalcostweight=1$)}
		\label{moo:graph-global-cost-energy}
	\end{subfigure}
	\hfill
	\begin{subfigure}{0.32\columnwidth}
		\includegraphics[width=\linewidth]{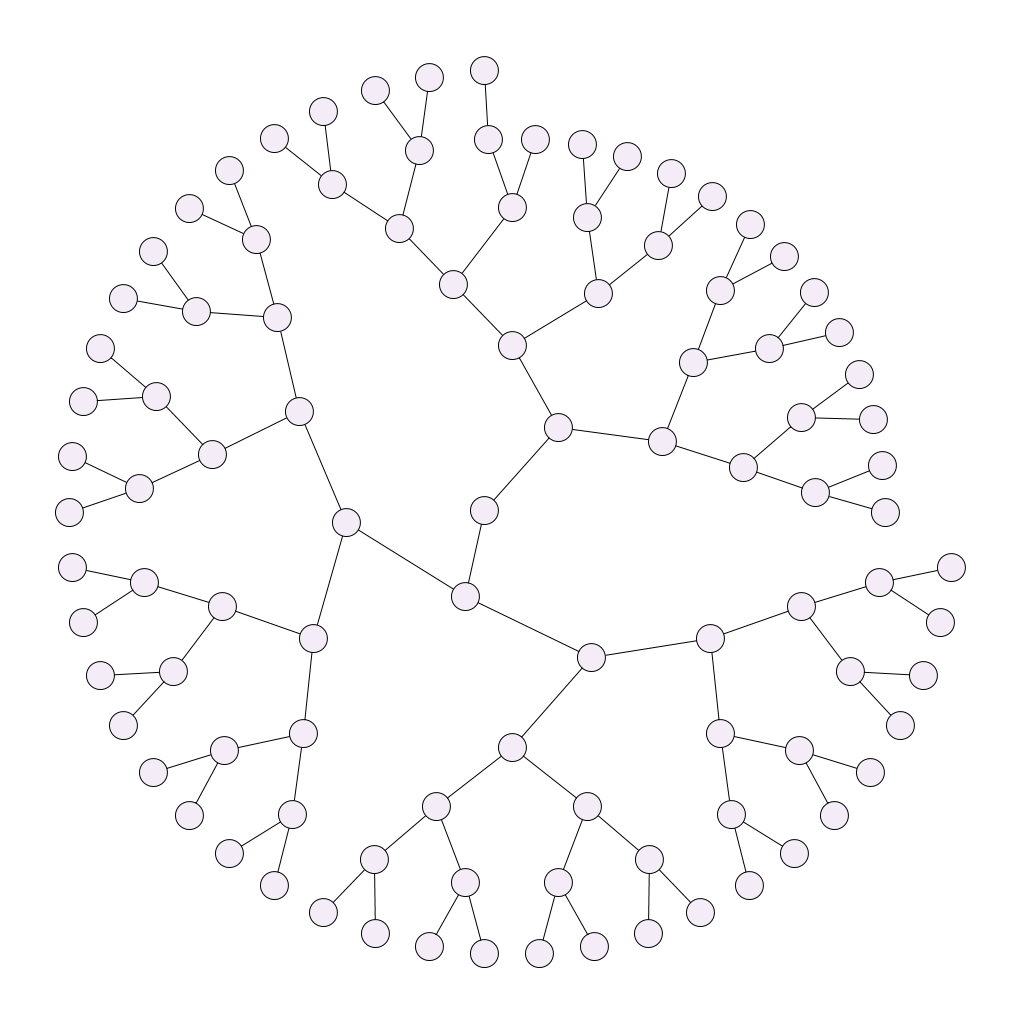}
		\caption{Comfort ($\localcostweight=1$)}
		\label{moo:graph-local-cost-energy}
	\end{subfigure}
	\hfill
	\begin{subfigure}{0.32\columnwidth}
		\includegraphics[width=\linewidth]{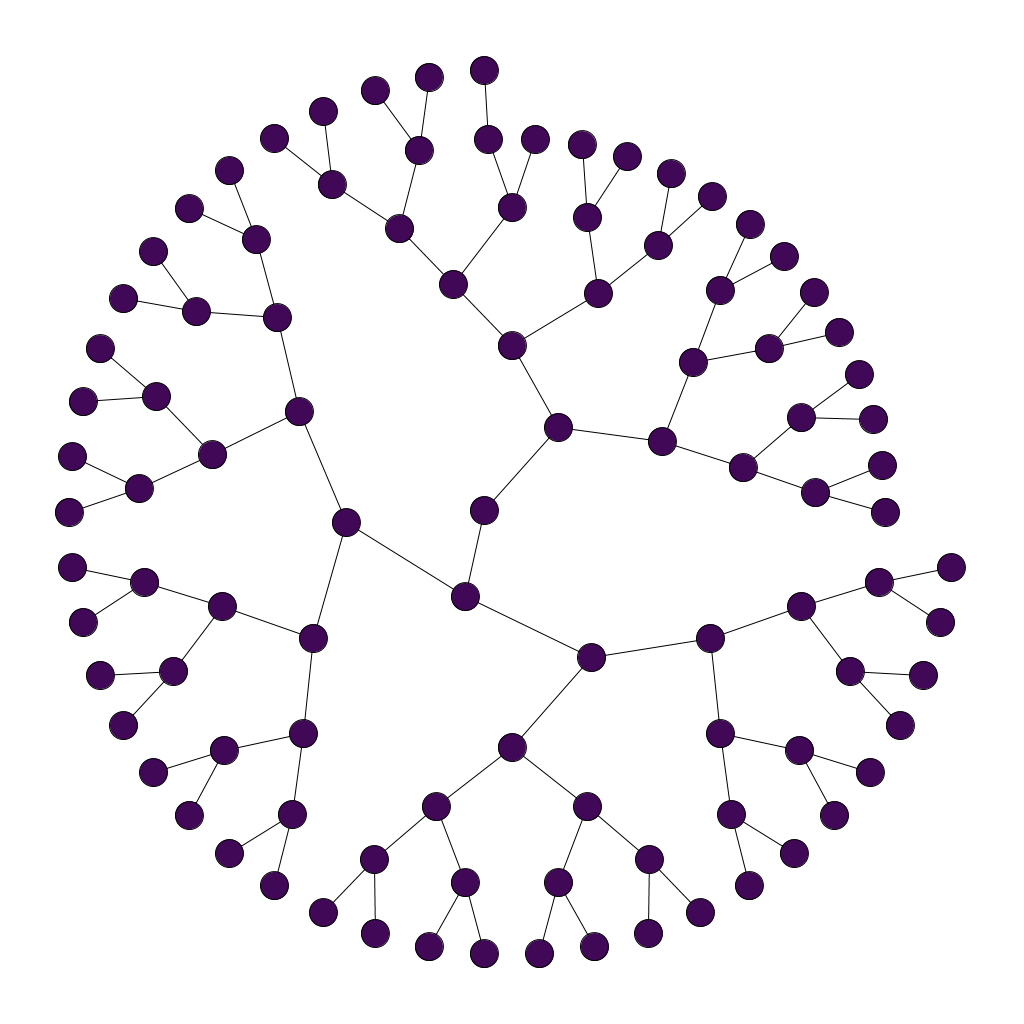}
		\caption{Fairness ($\unfairnessweight=1$)}
		\label{moo:graph-unfairness-energy}
	\end{subfigure}
	
	\caption{Discomfort (local) cost within the tree topology at the last optimization iteration and for different datasets.}
	\label{moo:graph-all-datasets}
	
\end{figure}

The Pareto frontier and the set of Pareto non-optimal solutions for all three datasets are visualized in Figure~\ref{moo:pareto-region}. Each solution is represented as a point in a 3-dimensional space, with each axis corresponding to one objective. The spread of values and the respective Pareto frontier provide insights of the trade-off between those objectives. These values are generated from the grid search of the weights as shown in Figure~\ref{moo:all-datasets}. The costs of the solutions for each parameter combination are shown along with the marked Pareto frontier solutions. 

\begin{figure*}[!htb]
	\centering	
	\subfloat[Synthetic]{%
		\includegraphics[width=0.325\linewidth]{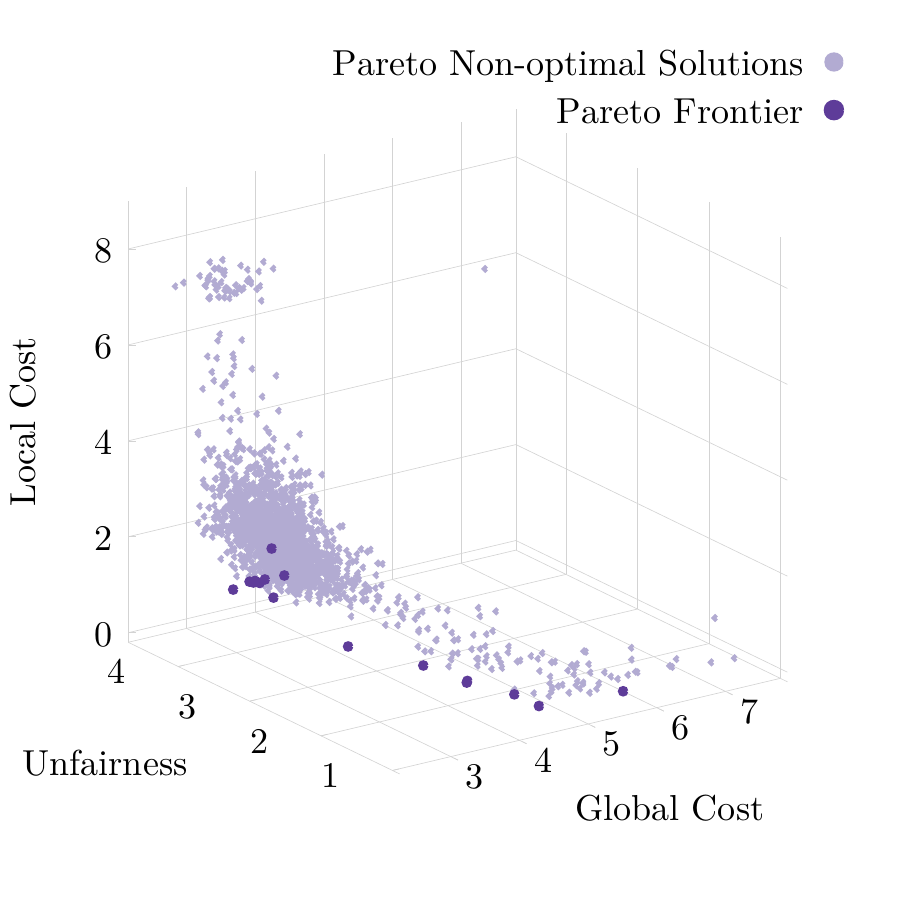}
		\label{moo:pareto-region-gaussian}}\hfill
	\subfloat[Bicycle]{%
		\includegraphics[width=0.325\linewidth]{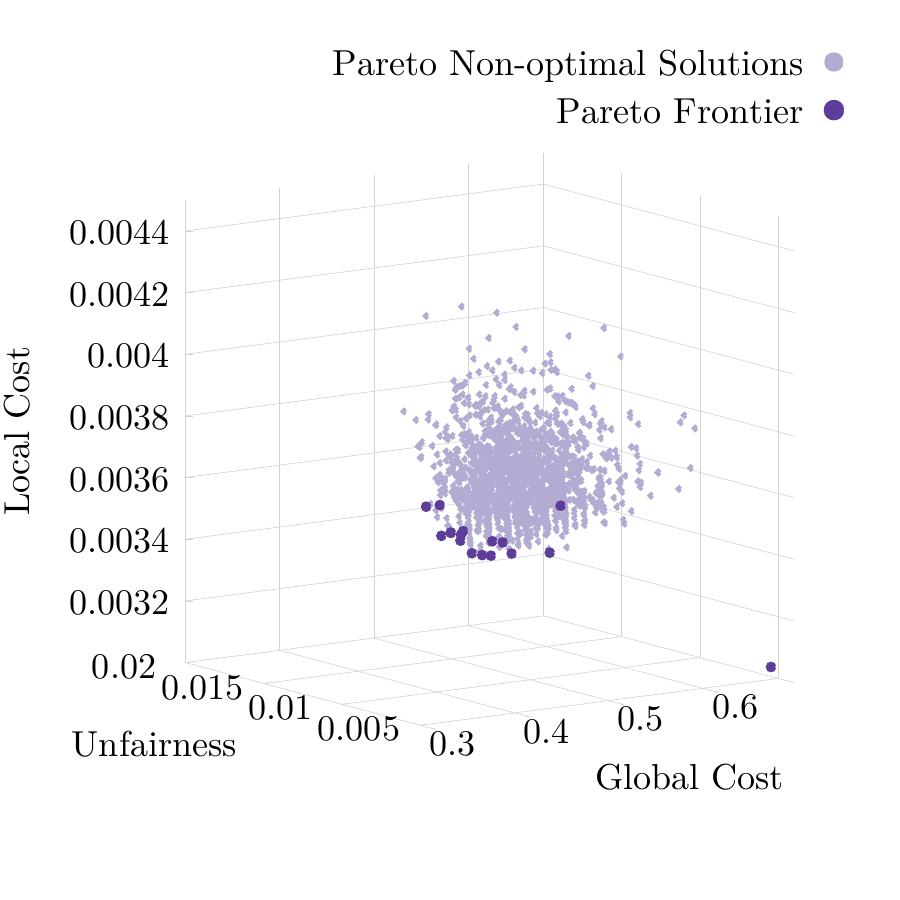}
		\label{moo:pareto-region-bicycle}}\hfill
	\subfloat[Energy]{%
		\includegraphics[width=0.325\linewidth]{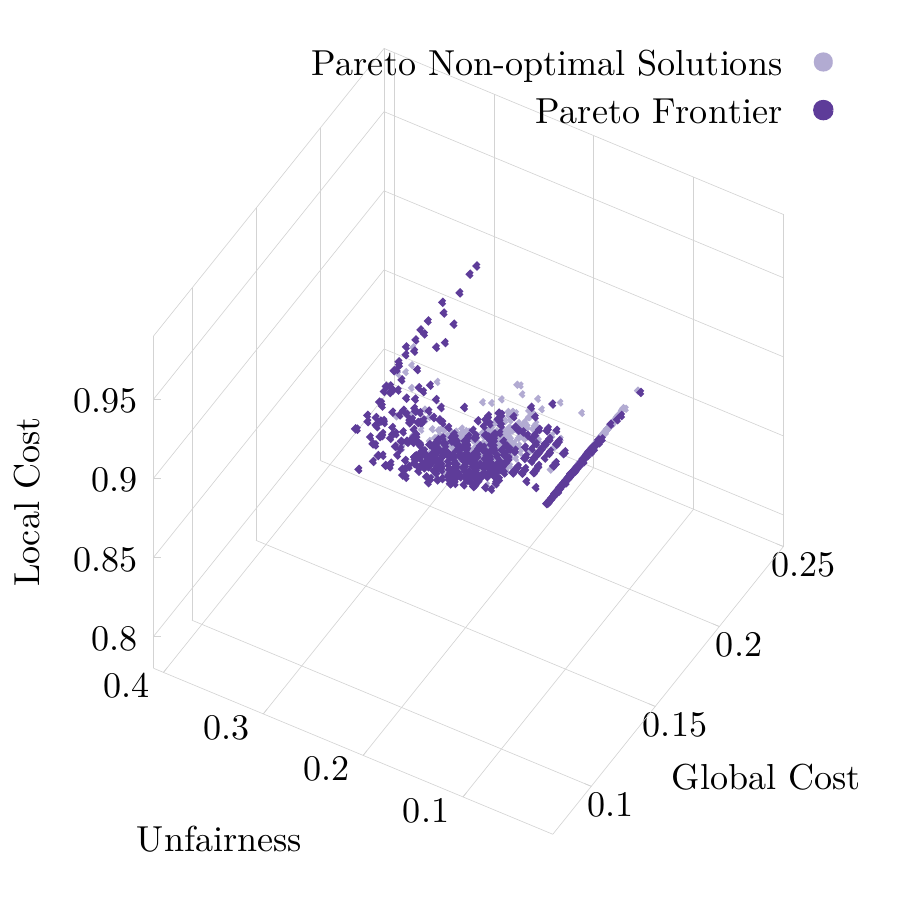}
		\label{moo:pareto-region-energy}}\hfill
	\caption{The feasible regions visualized for the (a) synthetic, (b) bicycle and (c) energy datasets. Inefficiency (global), discomfort (local) and unfairness cost are the three dimensions of the optimization space.}	
	\label{moo:pareto-region} 	
\end{figure*}

\begin{figure}[!htb]
	\centering
	\textbf{Synthetic}\\
	\subfloat[Inefficiency cost]{%
		\includegraphics[width=0.315\columnwidth]{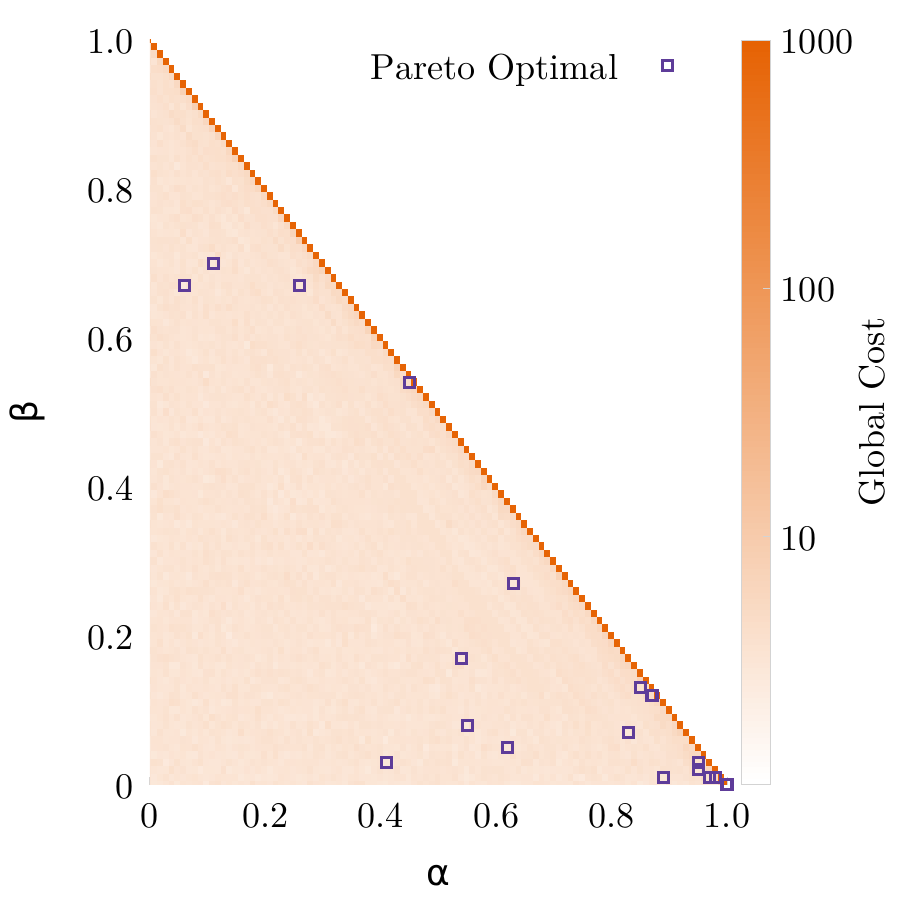}
		\label{moo:all-global-cost-gaussian}}
	\hfill
	\subfloat[Discomfort cost]{%
		\includegraphics[width=0.315\columnwidth]{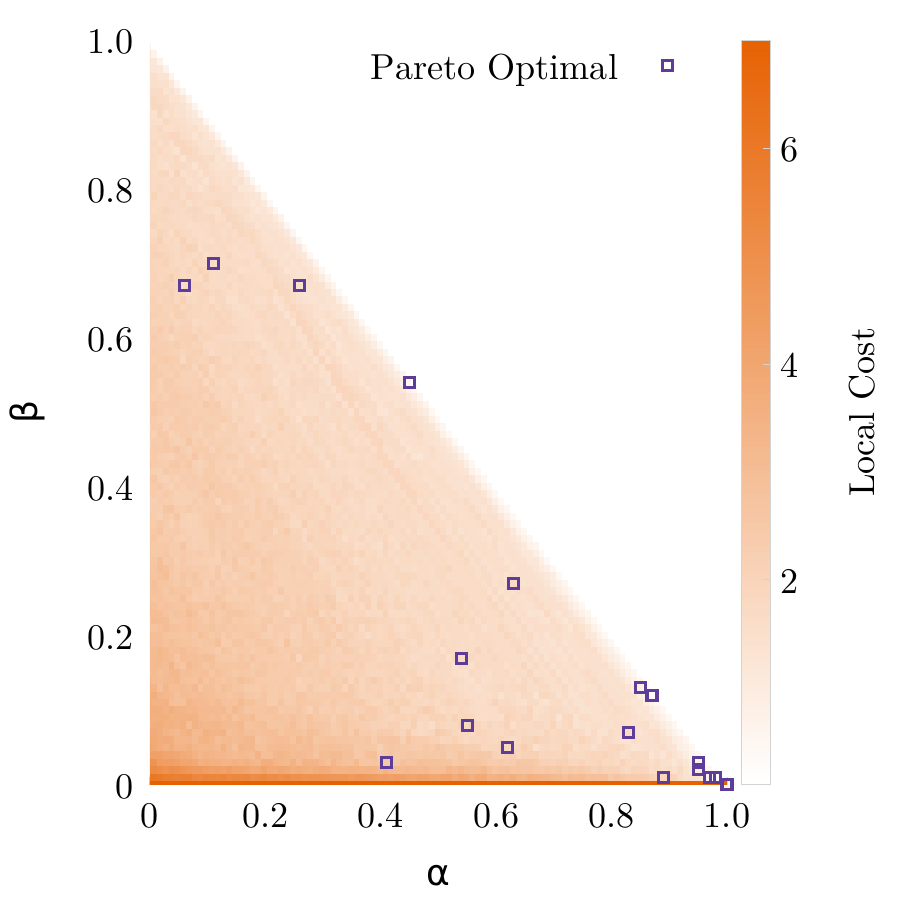}
		\label{moo:all-local-cost-gaussian}}
	\hfill
	\subfloat[Unfairness cost]{%
		\includegraphics[width=0.315\columnwidth]{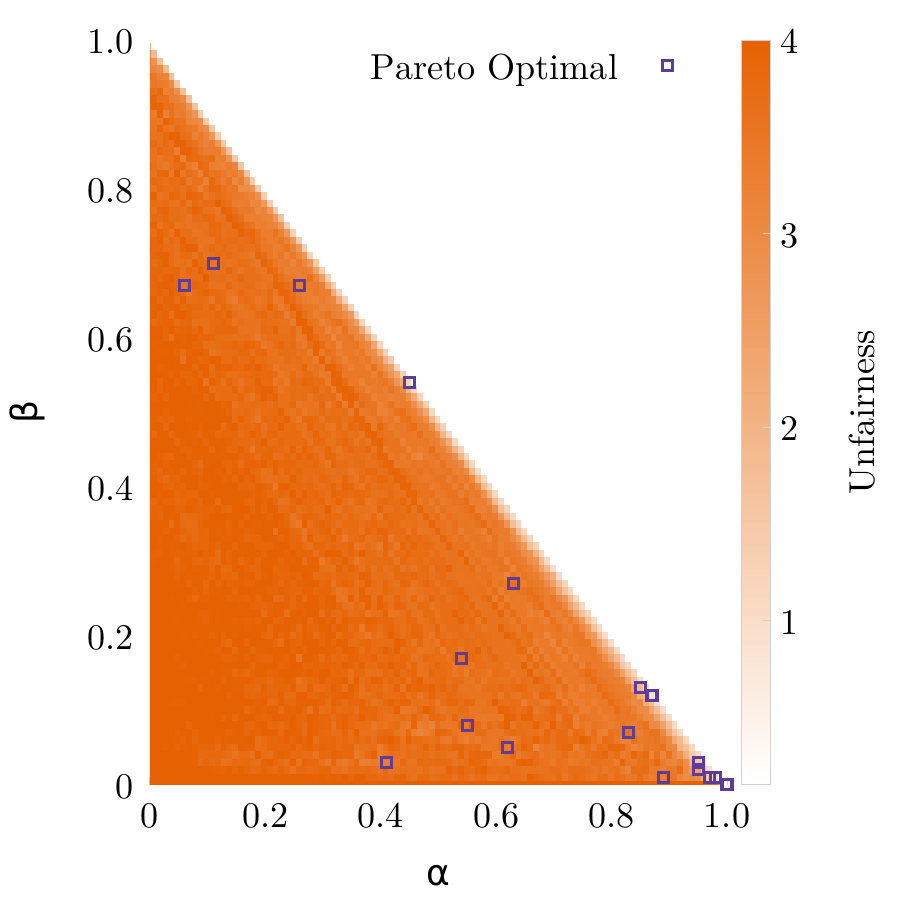}
		\label{moo:all-unfairness-gaussian}}\\
	\textbf{Bicycle}\\
	\subfloat[Inefficiency cost]{%
		\includegraphics[width=0.315\columnwidth]{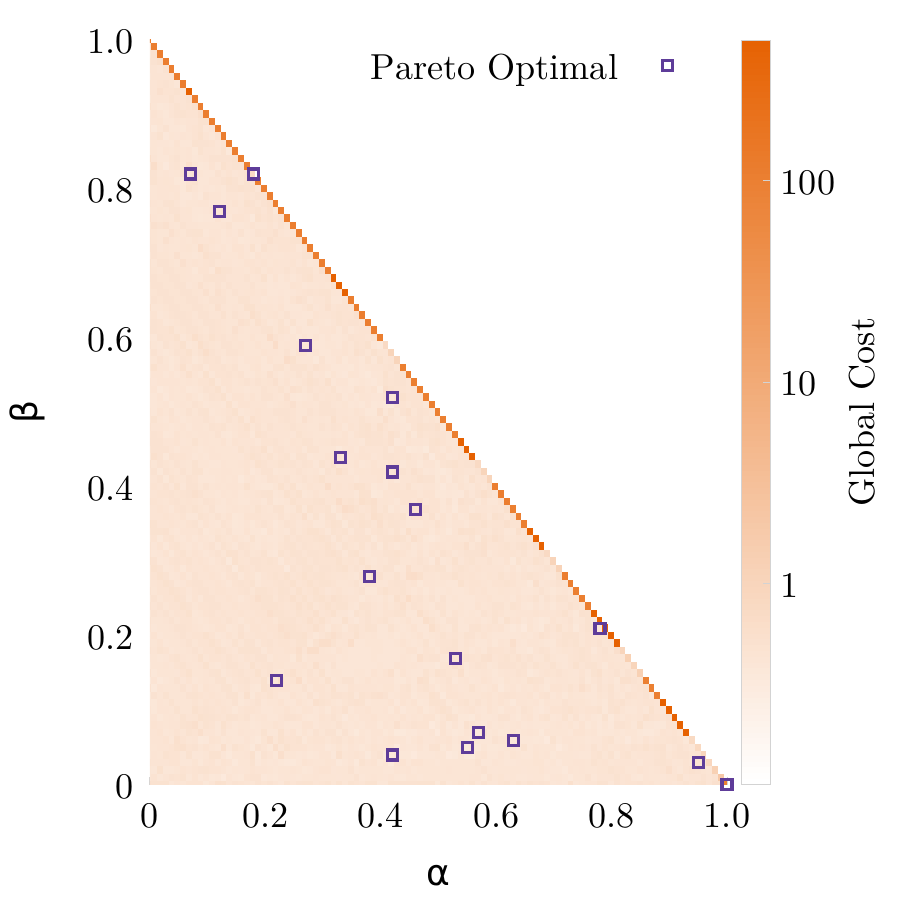}
		\label{moo:all-global-cost-bicycle}}
	\hfill
	\subfloat[Discomfort cost]{%
		\includegraphics[width=0.315\columnwidth]{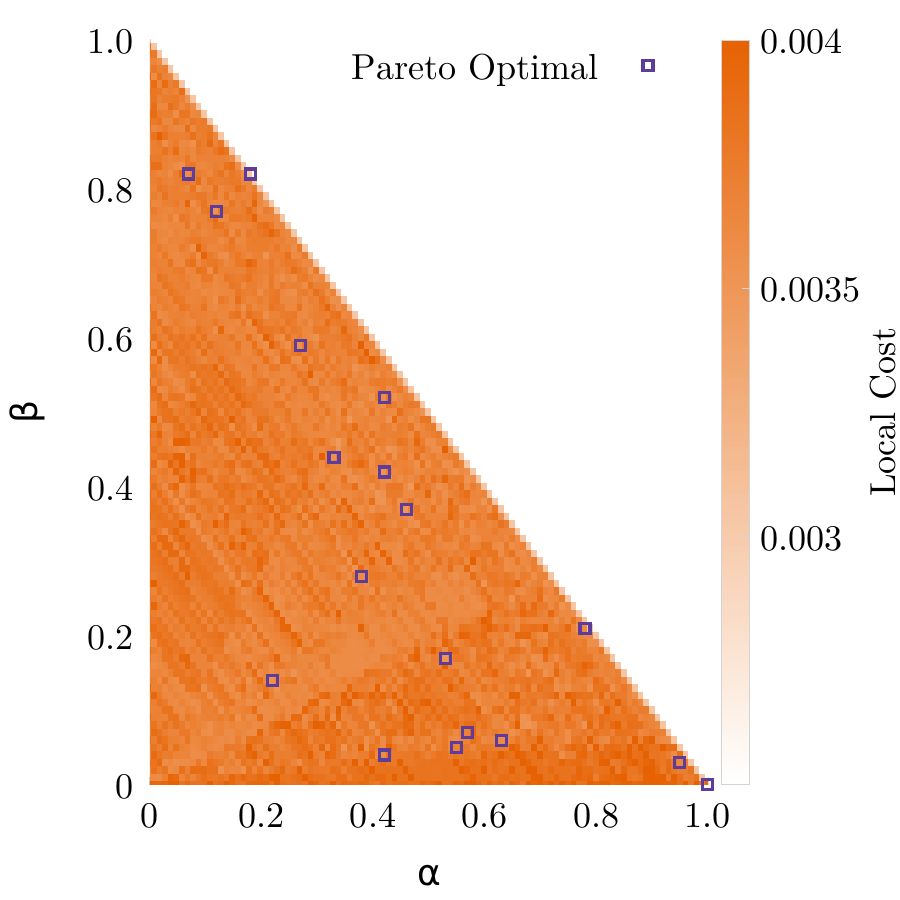}
		\label{moo:all-local-cost-bicycle}}
	\hfill
	\subfloat[Unfairness cost]{%
		\includegraphics[width=0.315\columnwidth]{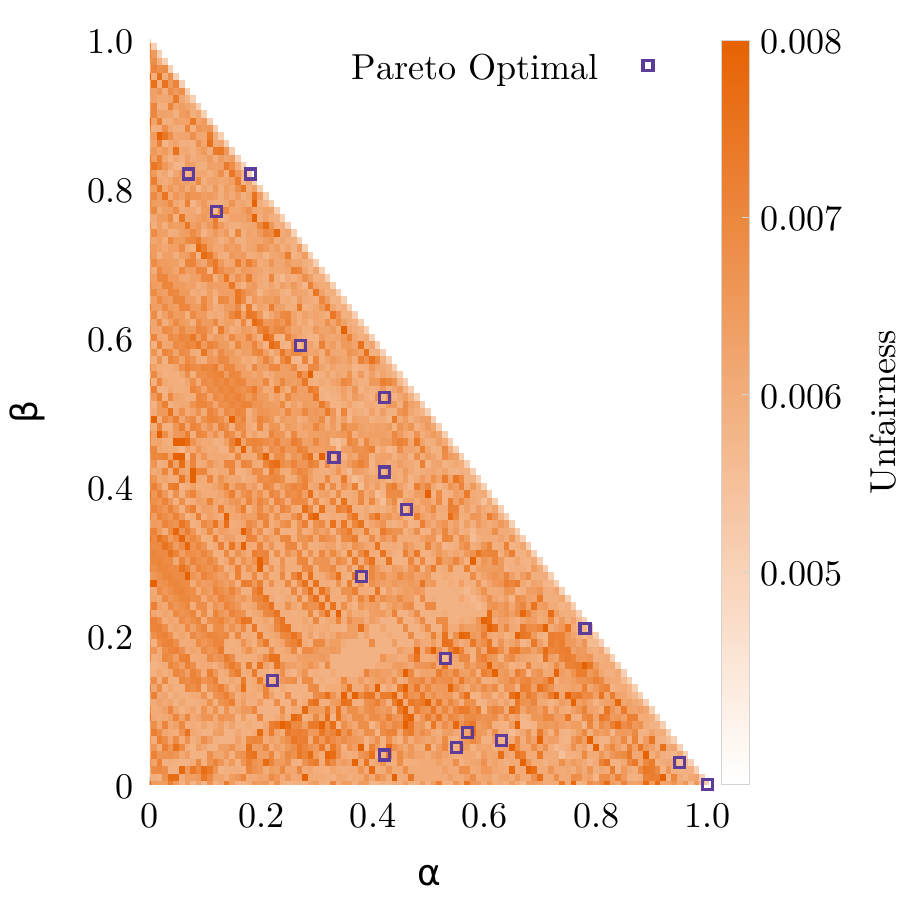}
		\label{moo:all-unfairness-bicycle}}\\
	\textbf{Energy}\\
	\subfloat[Inefficiency cost]{%
		\includegraphics[width=0.315\columnwidth]{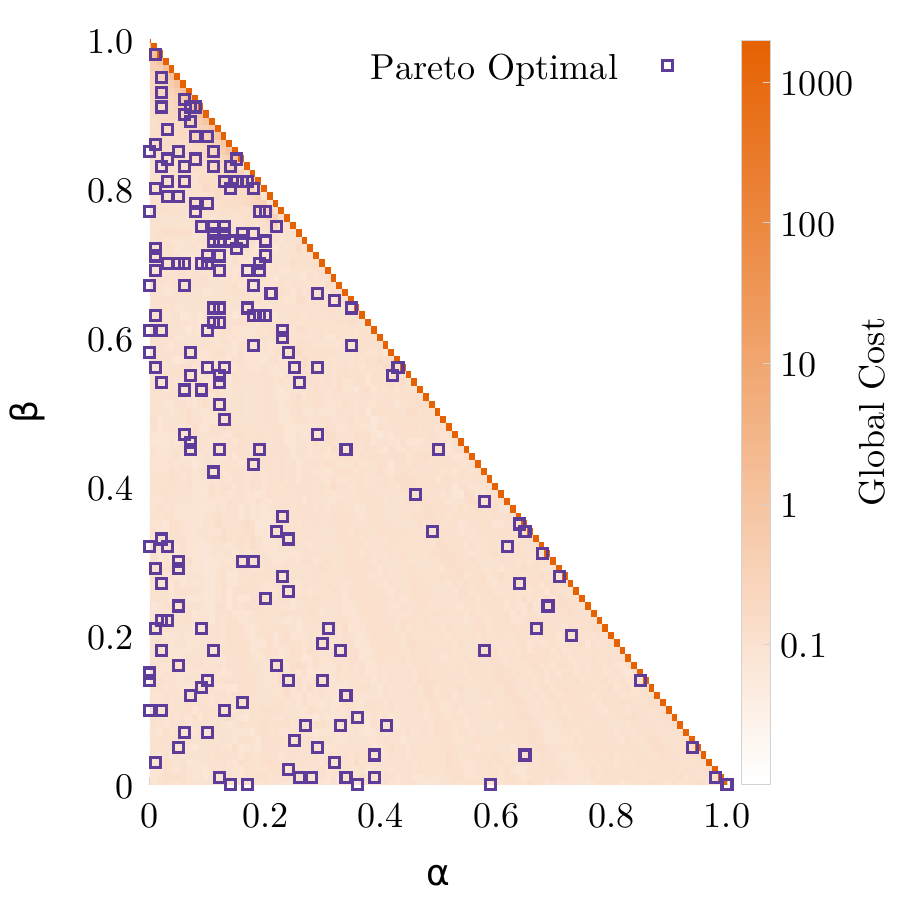}
		\label{moo:all-global-cost-energy}}
	\hfill
	\subfloat[Discomfort cost]{%
		\includegraphics[width=0.315\columnwidth]{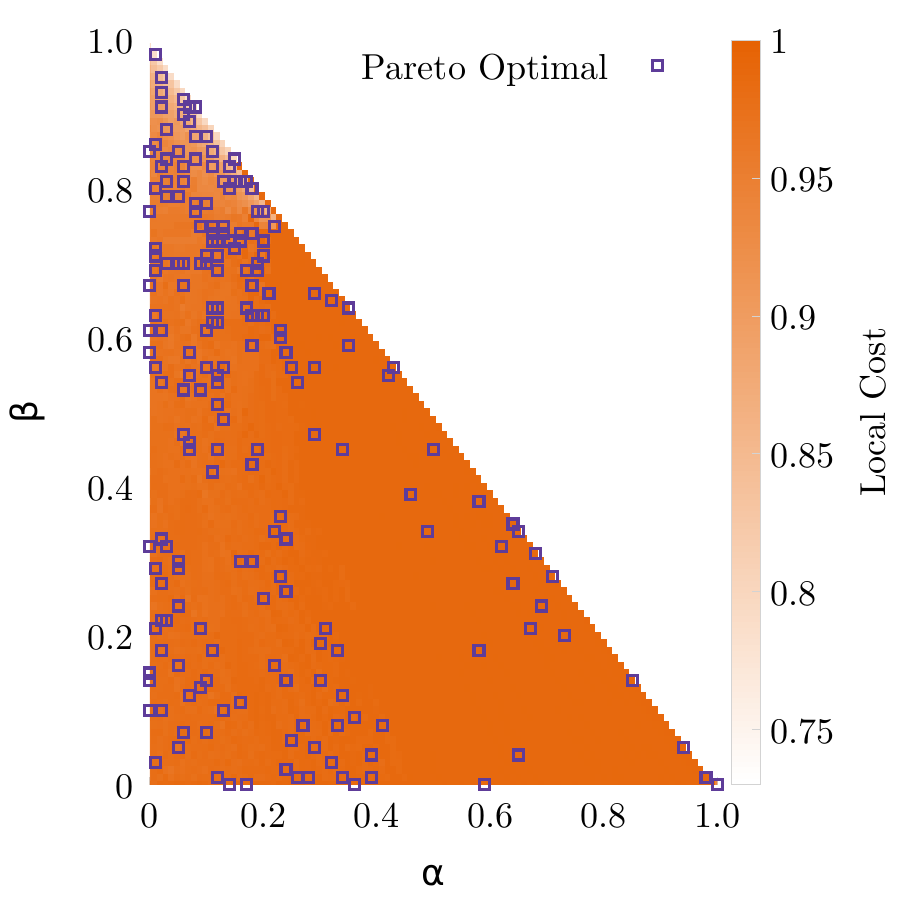}
		\label{moo:all-local-cost-energy}}
	\hfill
	\subfloat[Unfairness cost]{%
		\includegraphics[width=0.315\columnwidth]{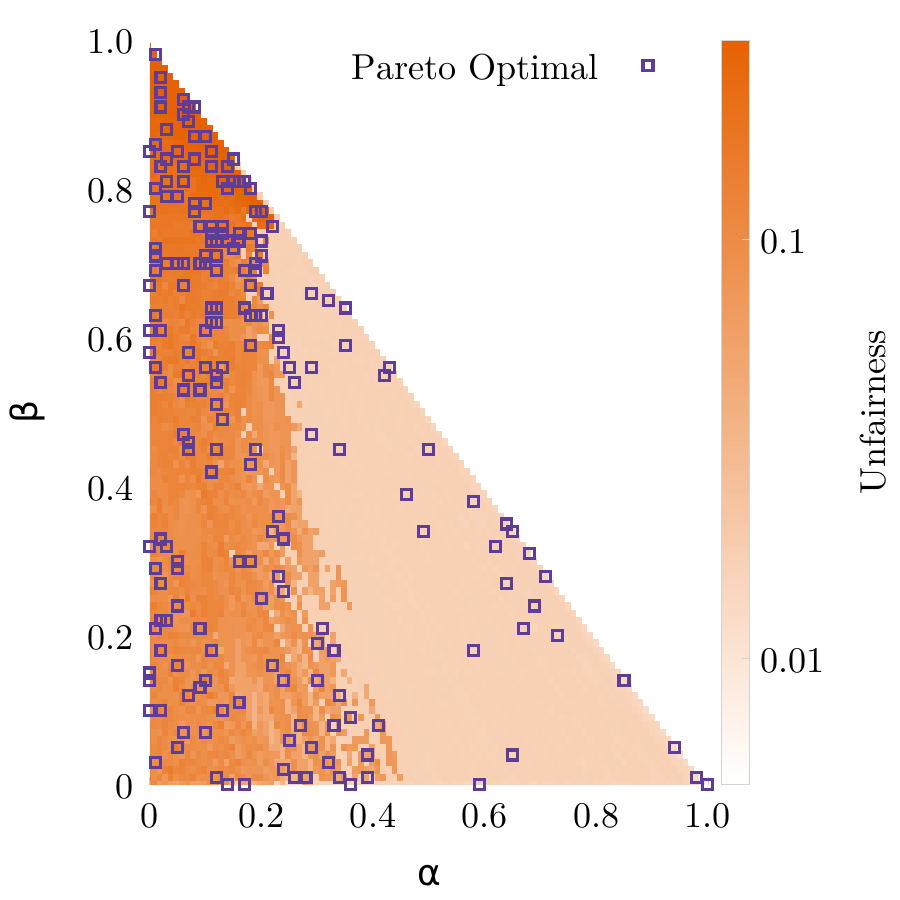}
		\label{moo:all-unfairness-energy}}\\
	
	\caption{Feasible regions of the three datasets with respect to the three objectives.}
	\label{moo:all-datasets}
\end{figure}

For the bicycle dataset (Figure~\ref{moo:pareto-region-bicycle} and Figure~\ref{moo:all-datasets}d-\ref{moo:all-datasets}f), the solutions are densely concentrated in one region with 17 solutions (0.33\%) comprising the Pareto frontier. On average, the spread between different solutions is limited across all three axes, with the majority of solutions grouped in one region of the Pareto frontier. A similar observation is made for the energy dataset (Figure~\ref{moo:pareto-region-energy} and Figure~\ref{moo:all-datasets}g-\ref{moo:all-datasets}i). However, the variation in unfairness across the frontier is far more pronounced, as the Pareto front extends primarily along the unfairness axis. Finally, the synthetic dataset (Figure~\ref{moo:pareto-region-gaussian} and Figure~\ref{moo:all-datasets}a-\ref{moo:all-datasets}c) illustrates a trade-off among the three objectives. The goal of employing synthetic data is to visualize such a trade-off in a controlled environment. For the observed orthogonality of the three objectives, the Pareto frontier is spread across the unfairness and inefficiency cost objectives.

\subsection{Discussion}\label{Experiments:Discussion}

The results allow us to draw a few fundamental conclusions. As shown and discussed in Section~\ref{Experiments:Results}, all three cost functions converge in no more than 35 iterations. This proves the fact that it is possible to minimize the inefficiency and discomfort costs in a collaborative (decentralized) manner, while keeping the system balanced, without the need to place significant disproportional burden on specific agents. Results are robust against threats to validity as shown in Section~\ref{app:validity}. 

The second observation can be made regarding the trade-off between orthogonal objectives. 	Depending on the characteristics of the dataset, such as the magnitudes of the costs and the values in the plans, optimization may be more sensitive to certain values of the weights. In the three available datasets, the trade-off between the objectives can be best controlled by fine-tuning the values of the weights in the range [0.9, 1]. 

A strong polarization in the inefficiency cost values is noticeable in all three datasets, and occurs for the same set of parameter triplets. The reason for this is the different orders of magnitude for the unfairness and discomfort cost compared to the inefficiency cost. Consequently, the inefficiency cost is the most dominant term that governs minimization, but as $\globalcostweight \rightarrow 0$, the other two terms gain dominance. This shows the need of the introduced fairness model, i.e. how optimization that takes into consideration only system-wide efficiency and agents' comfort may lead to a distribution of discomfort among agents that is heavily skewed towards a small subset of the agents, creating a strong system polarization. The findings also empirically showcase how the proposed approach can mitigate against uneven distribution of discomfort between agents and how the unfairness weight $\alpha$ influences this distribution.

\section{Conclusion and Future Work}\label{Section:Conclusions}

This article illustrates a model for adaptive optimization of three orthogonal objectives in a collaborative manner. Although the calculation of unfairness requires knowledge of the average discomfort cost of the agents' choice, we prove that it is possible to minimize unfairness using approximations, hence reusing information already available without adding unnecessary communication bottleneck. By applying the model to I-EPOS as an exemplary algorithm, we demonstrate that it is possible to simultaneously optimize system-wide efficiency, individuals' comfort, and the equitable distribution of discomfort costs in a fully decentralized way.

Via experiments on one synthetic and two real-world datasets, we show that all three objectives converge reliably within fewer than 35 iterations. Unfairness minimization prevents selective overburdening of individual agents. This is a prominent behavior in the bicycle dataset, where optimizing for system-wide efficiency alone causes disproportionate discomfort for a small subset of agents, creating risks for system operations. The Pareto frontier analysis confirms that the three objectives are indeed orthogonal across all datasets. Fine-tuning the objective weights is key to control trade-offs. 

Future work will study alternative fairness measures, heterogeneous prioritization of objectives within the population of agents, and how fairer optimization solutions can incentivize more altruistic behavior, meaning agents that are ready to accept higher discomfort values to improve efficiency. 

\section*{Acknowledgment}
This research is supported by a UKRI Future Leaders Fellowship (MR-/W009560/1): \emph{Digitally Assisted Collective Governance of Smart City Commons--ARTIO} and the EPSRC IAA project \emph{Collective Learning Optimisation Algorithm}.

\FloatBarrier
\clearpage

\bibliographystyle{IEEEtran}
\bibliography{references}


\appendices

\section{Other Fairness Measures}
\label{Appendix:Fairness}

\subsection{Jain's Fairness Index}
Raj Jain's fairness index \cite{sediq_optimal_2012} was developed to determine fair resource allocation in shared systems. It can be expressed directly in terms of the aggregated quantities defined in Section III as  $J : \mathbb{R}^d_+ \to [\frac{1}{N}, 1]$:
  \[
      J^{(t)} = \frac{\left(k_a^{(t)}\right)^2}{N \cdot K_a^{(t)}}
  \]
  where $K_a^{(t)} = \sum_{c \in \mathcal{T}_a} (l(s_c^{(t)}))^2$ and
  $k_a^{(t)} = \sum_{c \in \mathcal{T}_a} l(s_c^{(t)})$. A value of $J^{(t)} = 1$
  indicates perfect fairness, while $J^{(t)} = \frac{1}{N}$ indicates maximum
  unfairness.

  \begin{lemma}[Consistency of Jain's Fairness Index]
  Under the same conditions as Corollary~1, with $\mathbb{E}[X^2] > 0$,
  \[
      J^{(t)} \xrightarrow{a.s.} \frac{(\mathbb{E}[X])^2}{\mathbb{E}[X^2]}
      \quad \text{as } N \to \infty.
  \]
  \end{lemma}

  \begin{proof}
  By definition,
  \[
      J^{(t)} = \frac{\left(\frac{1}{N}k_a^{(t)}\right)^2}{\frac{1}{N}K_a^{(t)}}.
  \]
  By the strong law of large numbers, $\frac{1}{N}k_a^{(t)} \xrightarrow{a.s.}
  \mathbb{E}[X]$ and $\frac{1}{N}K_a^{(t)} \xrightarrow{a.s.} \mathbb{E}[X^2]$.
  Since $\mathbb{E}[X^2] > 0$ by assumption, the denominator is bounded away
  from zero almost surely, and the result follows by the continuous mapping theorem.
  \end{proof}

\section{Datasets and Threats to Validity}
\label{Appendx:Datasets}
\subsection{Energy}\label{Appendix:Energy}
This dataset was obtained by disaggregation~\cite{pournaras2017} of the \textit{simulated zonal power transmission in the Pacific Northwest} dataset. The resulting dataset contains power loads of 5600 users. Each agent has exactly 10 possible plans, each of length 144. The first is the original disaggregated load; the next 3 are obtained by applying the SHUFFLE generation scheme to the original; the next 3 are obtained by applying the SWAP-15 generation scheme; and the final 3 are obtained by applying the SWAP-30 generation scheme. The first plan is the most preferred, as it is the only real-world plan and therefore has a preference score of $1.0$. The SHUFFLE generation scheme randomly shuffles the values in the first plan. Since this scheme perturbs all 144 values of the original plan, the SHUFFLE plans have the lowest preference score, equal to $\frac{1}{144}$. The SWAP-15 generation scheme randomly selects 15 pairs of values of the original plan to swap, and their preference is adjusted to $\frac{1}{15}$. Similarly, the SWAP-30 scheme randomly selects 30 pairs to swap, and their preference score is $\frac{1}{30}$. Consequently, the mean and the standard deviation of the possible plans are constant for each agent.

The preference scores are used for discomfort cost optimization, which entails selecting plans with high preference scores relative to the other efficinecy objectives. As I-EPOS is a minimization algorithm, preference scores are converted to their respective \textit{complements} called \textit{costs} as follows: if the preference score is $h$, then its cost is $1-h$. Note that each cost value is between 0 and 1. 

\subsection{Bicycle}\label{Appendix:Bicycle}
This dataset was generated from the data obtained from the \textit{Hubway Data Visualization Challenge}\footnote{Hubway Data Challenge, Hubway, Accessible at: http://hubwaydatachallenge.org. Last accessed in July 2026.} that consists of the trip records of the Hubway bike-sharing system in Paris. The original data do not contain user identification, but include more general user information, such as zip code, year of birth, and gender. All records with the same values of these three identification keys are assumed to be one user, and a subset of 1000 such users is included in the dataset. A possible plan is a sequence of values, one for each bike station, that indicates the difference between the numbers of incoming and outgoing trips at that station. For example, if a user cycled between stations 1 and 3, the corresponding possible plan is $(-1, 0, 1, 0, ...)$. Different trips for a single user are encoded as distinct possible plans. Finally, each possible plan (a trip) is assigned a \textit{score} indicating the likelihood that a user does not make the trip. Note that scores of all plans in the whole dataset belong to the range $[0, 1]$.

\subsection{Synthetic}\label{Appendix:Synthetic}
The values in this dataset are drawn from a Normal distribution. The dataset consists of 1000 agents, each having 16 possible plans. The possible plan is a vector of length 100. Each possible plan is coupled with a score that indicates the index of the possible plan, and this score is used as the discomfort cost of the plan. In other words, the first plan has the score of 0, the second one the score of 1 and so on. It is used as a proof of concept since it is the most flexible dataset.

\subsection{Additional Experimental Details}\label{Appendix:Experiments}

\begin{table}[!htb]
	\centering
	\caption{Experimental settings for obtaining the feasible region as presented in Section \ref{Section:Experiments}.}
	\label{table:feasible-region-settings}
	\resizebox{\columnwidth}{!}{
    \begin{tabular}{l l l l l}
		\toprule
		\multirow{2}{*}{Parameters} && \multicolumn{3}{c}{\textit{Datasets}} \\
		\cmidrule{3-5}
		&&	Bicycle	& Energy & Synthetic \\
		\midrule
		Efficnecy objective \globalcostfunc &&	\multicolumn{3}{c}{Variance} \\
		Weights && \multicolumn{3}{c}{\makecell[c]{\globalcostweight, \unfairnessweight and \localcostweight are varied between 0 and 1 with step 0.01, \\ $\globalcostweight + \unfairnessweight + \localcostweight = 1$, \globalcostweight, \unfairnessweight and \localcostweight are constant over time and for all agents}} \\
		\midrule		
		Number of agents	&&	$|N| = 1000$ &	$|N| = 1000$	&	$|N| = 1000$	\\
		Number of plans		&&	$|\mathcal{P}_{a}| \leq 23, \forall a \in \mathcal{A}$ &	$|\mathcal{P}_{a}| = 10, \forall a \in \mathcal{A}$ &	$|\mathcal{P}_{a}| = 16, \forall a \in \mathcal{A}$	\\
		Dimensions of plans	&&	$d= 98$ &	$d = 144$ &	$d = 100$	\\
		Number of children && $n = 2$ & $n = 2$ & $n = 2$ \\
		Number of iterations	&&	$T = 40$	&	$T = 40$ &	$T = 40$ \\		
		\bottomrule
	\end{tabular}
    }
\end{table}

\subsection{Threats to Validity}\label{app:validity}

All experiments are repeated 1000 times with different initial relative positions of the agents in the hierarchy. This means that the order of coordination is controlled as a factor of influencing the performance. Experimental evidence is provided for for both synthetic and real-world data. The influence of other factors such as the number of children per agent, the number of plans and size of plans are studied in earlier work~\cite{pournaras2018,Nikolic2019} and are not expected to alter the conclusions of the paper. 

\end{document}